\documentclass[5p,twocolumn]{elsarticle}
\usepackage{flushend}
\usepackage{subfigure}
\usepackage{graphicx,color}
\usepackage{latexsym}
\usepackage{amsmath} 
\usepackage{amssymb}  
\usepackage{amsmath,amssymb,amsfonts,theorem}
\usepackage{pgf}
\usepackage{tikz}
\usepackage{booktabs}
\usepackage{pifont}
\newcommand*\rot{\rotatebox{90}}

\usetikzlibrary{arrows,automata}
\usepackage{enumerate}

{ \theorembodyfont{\normalfont} 

\def\startmodif{\color{black}}

\def\smod{\color{black}}
\def\emod{\color{black}} 

\def\smodc{\color{black}}
\def\emodc{\color{black}} 

\def\startmodif{}
\def\stopmodif{\color{black}}

\newcommand{\R}{\mathbb{R}}

\newcommand{\mbf}{\mathbf}
\newcommand{\trans}{T}

\newcommand{\allones}{\mathbf{1}_n}

{
\newtheorem{example}{Example}
\newtheorem{rem}{Remark}
\newtheorem{assum}{Assumption}

\newtheorem{thm}{Theorem}

\newtheorem{lem}{Lemma}
\newtheorem{corollary}{Corollary}
\newtheorem{proposition}{Proposition}

\newcommand{\dst}{\displaystyle}
\def\be{\begin{equation}}
\def\ee{\end{equation}}
\def\ba{\begin{array}}
\def\ea{\end{array}}
\def\eqa{\begin{eqnarray}}
\def\eqe{\end{eqnarray}}

\begin{document}
\begin{frontmatter}
\title{An internal model approach to \smod (optimal) \emod frequency regulation in power grids \smod with time-varying voltages\emod\tnoteref{t1}} 

\tnotetext[t1]{The work of S. Trip and C. De Persis is \smodc supported  by the Danish Council for Strategic Research (contract no. 11-116843) within the `Programme Sustainable Energy and Environment', under the “EDGE” (Efficient Distribution of Green Energy) research project. \emodc The work of  C.~De Persis is also supported by {\it QUICK} (The Netherlands Organization of Scientific Research) and {\it Flexiheat} (Ministerie van Economische Zaken, Landbouw en Innovatie). The work of M.~B\"{u}rger was supported by the German Research Foundation (DFG) within the Cluster of Excellence in Simulation Technology (EXC 310/2) at the University of Stuttgart. The work was carried out while he was with the Institute for Systems Theory and Automatic Control, University of Stuttgart, Pfaffenwaldring 9, 70550 Stuttgart, Germany.  An abridged version of the paper has appeared in \cite{burger.et.al.mtns14b}.}
\author[RUG]{S. Trip\corref{cor1}}\ead{s.trip@rug.nl} 
\author[Mathias]{M. B\"{u}rger}\ead{mathias.buerger@de.bosch.com}    
\author[RUG]{C. De Persis}\ead{c.de.persis@rug.nl}         
\cortext[cor1]{Corresponding author. Tel.: +31 50 363 3077}
\address[RUG]{\smod ENTEG, \emod Faculty of Mathematics and Natural Sciences, University of Groningen, Nijenborgh 4, 9747 AG Groningen, the Netherlands}             
\address[Mathias]{Cognitive Systems, Corporate Research, Robert Bosch GmbH,
Robert Bosch Campus 1, 71272 Renningen, Germany}

\begin{abstract}                          
 This paper studies the problem of frequency regulation in power grids under unknown and possible time-varying load changes, while minimizing the generation costs. We formulate this problem as an output agreement problem for distribution networks and address it using incremental passivity and distributed internal-model-based controllers. \startmodif Incremental passivity enables a systematic approach to study convergence to the steady state with zero frequency deviation and to design the controller in the presence of time-varying voltages, whereas the internal-model principle is applied to tackle  the uncertain
nature of the loads. \stopmodif
\end{abstract}
\end{frontmatter}

\section{Introduction}
The power grid can be regarded {as} a large interconnected network of different subsystems, called control area's. In order to guarantee reliable operation, the frequency is tightly regulated around its nominal value, e.g. 60Hz. Automatic regulation of the frequency in power grids is traditionally achieved by primary proportional control (droop-control) and a secondary PI-control. In this secondary control, commonly known as automatic generation control (AGC), each control area determines its ``Area Control Error'' (ACE) and changes its production accordingly to compensate for local load changes in order to regulate the frequency back to its nominal value and to maintain the scheduled power flows between different area's.\\
By requiring each control area to compensate for their local load changes the possibility to achieve economic efficiency is lost. Indeed, the scheduled production in the different control area's is currently determined by economic criteria relatively long in advance. To be economically efficient an accurate prediction of load changes is necessary. Large scale introduction of volatile renewable energy sources and the use of electrical vehicles will however make accurate prediction difficult as the net load (demand minus renewable generation) will change on faster time scales and by larger amounts.\\
Due to the difficulty of precisely predicting the load, the problem of designing algorithms for power generation able to maintain the network at nominal operating conditions despite the effect of unmeasured power demand and while retaining economic efficiency has attracted considerable attention and a vast literature is already available. The aim of this paper is to provide a different framework in which the problem can be tackled exploiting the incremental passive nature of the dynamical system adopted to model the power network and internal-model-based controllers (\cite{Pavlov2008}, \cite{burger.depersis.aut13}) able to achieve an economically efficient power generation control in the presence of possibly time-varying power demand. We focus on a third-order model with time-varying voltages known as ``flux-decay model" (\cite{chiang.et.al.proc.ieee1995}, \cite{machowski_power_2008}), which, although simplistic, is tractable and meaningful.

\medskip

{\em Literature review.} An up-to-date review of current research on AGC can be found in \cite{Overview}. The economic efficiency of AGC has attracted considerable attention
and the vast literature available   makes the task of providing  an exhaustive survey very difficult.
Relevant results which are close to the present paper are briefly discussed below to better emphasize our contribution.\\
%
In \cite{KTH}, distributed and centralized controllers that require the knowledge of the frequency deviations at the bus and its neighbors are proposed for the linearized version of the swing equation and shown to achieve frequency regulation while minimizing a quadratic cost function under a suitable matrix condition. An economically efficient, discrete time AGC algorithm incorporating generator constraints is proposed in \cite{apos2014} and investigated numerically.
The use of distributed proportional and proportional-integral  controllers for microgrids has been studied in \cite{simpson-porco_synchronization_2013}, \cite{guerrero2011hierarchical}
 with additional economic insights provided in \cite{dorfler2014breaking}, where, among other contributions, decentralized tertiary control strategies have been proposed.
Investigation of stability conditions for droop controllers in a port-Hamiltonian framework and in the presence of time-varying voltages was pursued in \cite{schiffer_synchronization_????}. In \cite{zhang.papa.acc13}, \cite{li_connecting_2013},  the problem of optimal frequency regulation was tackled by formulating suitable  optimal power flow problems, characterizing their solutions and then providing gradient-like algorithms that asymptotically converge to the optimum. While \cite{zhang.papa.acc13} focused on power networks with star topology, quadratic cost functions  and including equality and inequality constraints, the paper  \cite{li_connecting_2013} does not assume any specific topology for the network and considers convex cost functions, but assumes the knowledge of the power flows at the buses to guarantee the achievement of the desired steady state solution. Work relating automatic generation control and optimal load control has appeared in \cite{zhao_power_2013}, \cite{zhao2014decentralized}, with the former focusing on linearized power flows and without generator-side control and the latter removing these assumptions. \\

{\em Main contribution.} The contribution of this paper is to propose a new approach to the problem
that differs substantially from the aforementioned works.
We move along the lines of \cite{burger.depersis.aut13}, \cite{Burger2013b}, where a framework to deal with nonlinear output agreement and optimal flow problems for dynamical networks has been proposed. In those papers
internal-model-based dynamic controllers have been designed to solve output agreement problems for networks of incrementally passive systems (\cite{Pavlov2008})
in the presence of time-varying perturbations. In this paper we build upon \cite{burger.et.al.mtns14b}. After showing that the dynamical model adopted to describe the power network is an incrementally passive system with respect to solutions that are of interest (solutions for which the frequency deviation is zero), we provide a systematic method to design internal-model-based power generation controllers that are able to balance power loads while minimizing the generation costs at steady state. \\
This design is carried out first by solving the regulator equations (\cite{Pavlov2008},\cite{burger.depersis.aut13}) associated with the frequency regulation problem. \smod Among the feedforward power generation inputs  that solve the  regulator equations, we single out the one for which the static optimal generation problem is solved. \emod Then, following  \cite{burger.depersis.aut13}, \cite{Burger2013b}, an internal-model-based incrementally passive  controller is proposed which is able  to generate in open-loop the desired feedforward input and stabilize the closed-loop system in such a way that all the solutions converge to the desired synchronous solution and to the optimal generation control. \\
Although the proposed incrementally passive  controllers share similarities with others presented in the literature, the way in which they are derived is to the best of our knowledge new. Moreover, they show a few advantages.
\\
(i)  If we allow for time-varying power demand in the model, our internal-model controllers can deal with this scenario and it turns out that  proportional-integral controllers that are more often found in the literature are a special instance of these controllers. \\
 (ii)  Being based on output regulation theory for systems over networks (\cite{burger.depersis.aut13}, \cite{Burger2013b}, \cite{Wieland2011}, \cite{Isidori2013}, \cite{DePersis2012a}), our approach has the potential to deal with fairly rich classes of external perturbations (\cite{cox2012}, \cite{Serrani2001}), thus paving the way towards frequency regulators in the presence of a large variety of consumption patterns. Furthermore, other extensions of \cite{burger.depersis.aut13} considered the presence of non-quadratic cost functions and flow capacity constraints (\cite{burger.et.al.mtns14}, \cite{burger.et.al.tcns14}) that could turn out to be useful also for the problem considered here. See \cite{li_connecting_2013}, \cite{zhao_power_2013}, \cite{zhao2014decentralized} for a different approach to deal with non-quadratic cost functions and constraints. \\
 (iii)  Passivity is an important feature shared by more accurate models of the power network, as already recognized in \cite{zonetti2012}, \cite{caliskan2014}, and in \cite{schiffer_synchronization_????} in the context of microgrids, implying that the methods that are employed  in this  paper might be used to deal with more complex (and more realistic) dynamical models. Although this level of generality is not pursued in this paper, the passivity framework allows us to include voltage  dynamics in our model, a feature that is usually neglected in other approaches (\cite{KTH}, \cite{zhang.papa.acc13}, \cite{li_connecting_2013}, but see \cite{schiffer_synchronization_????} for the inclusion of time-varying voltages in the case of microgrids, and also \cite{simpson-porco_synchronization_2013}).  Furthermore, passivity is a very powerful tool in the analysis and design of dynamical control networks
 (\cite{bai2011cooperative}, \cite{schaft2013}).\\
 (iv)  To show incremental passivity  we introduce storage functions that interestingly can be interpreted as energy functions, thus establishing a connection with classical work in the field (see e.g.~\cite{bergenTPAS81}, \cite{chiang.et.al.proc.ieee1995} and references therein) that can guide a further investigation of the problem. For instance, it can lead to the inclusion of automatic voltage regulators (\cite{chiang.et.al.proc.ieee1995}, \cite{miyagi1104229}) in the analysis, a study that is not explored in this paper.
\vspace{1em}\\
The paper is organized as follows. In Section \ref{sec.2}, we introduce the dynamical model adopted to describe the power grid. In Section \ref{sec.3}, we analyze the dynamical model assuming constant generation, and show that it leads to a nonzero frequency deviation. In Section \ref{sec.4}, we characterize the optimum generation to minimize the generation costs. { In Section \ref{sec.contr}, we propose a distributed controller which ensures frequency regulation and at the same time minimizes the generation costs under the assumption of constant demand. In Section \ref{sec.6}, the restriction of constant demand is relaxed and we extend results to the case of a certain class of time-varying demands.} In Section \ref{sec.7}, we test our controllers for an academic case study using simulations. In Section \ref{sec.8}, conclusions are given and an outline for future research is provided.

\section{System model}\label{sec.2}
The history of power grid modelling is rich and the models we adopt
can be found in most textbooks on power systems such as \cite{machowski_power_2008}.
We focus on swing equations to take into account the
frequency dynamics and, comparing to recent work which also concerns optimal frequency regulation \cite{burger.et.al.mtns14b}, \cite{dorfler2014breaking},  \cite{li_connecting_2013}, \cite{zhao2014decentralized},  we do not assume constant voltages.
We rather use an extended model which captures essential voltage dynamics \cite{chiang.et.al.proc.ieee1995},
 including constant voltages as a particular case, and possesses some incremental passivity properties
that are essential to our approach to the problem. In this work we assume that the power grid is partitioned into smaller areas, such as control areas, where the dynamic behavior of an area can be described by an equivalent single generator as a result of coherency and aggregation techniques \cite{chakrabortty2011measurement}, \cite{Ourari2006}.
As a consequence we do not distinguish between individual generator and load buses, similar to the work in \cite{KTH}, \cite{li_connecting_2013} and \cite{Zhao2014}.
This is in contrast with the structure-preserving models in e.g. \cite{chiang.et.al.proc.ieee1995} and \cite{bergenTPAS81}, where the load buses are explicitly modelled or with Kron-reduced models in e.g. \cite{schiffer_synchronization_????} and \cite{dorflersiam}, where load buses can be eliminated by modeling them as constant admittances.
\vspace{1em}\\
Consider a power grid consisting of $n$ areas. The network is represented by a connected and undirected graph $\mathcal{G} = (\mathcal{V}, \mathcal{E})$, where the nodes, $\mathcal{V} = \{1, \hdots, n\}$, represent control areas and the edges, $\mathcal{E} \subset \mathcal{V} \times \mathcal{V} = \{1,\hdots, m\}$, represent the  transmission lines connecting the areas. The network structure can be represented by its corresponding incidence matrix $D \in \mathbb{R}^{n \times m}$. The ends of edge $k$ are arbitrary labeled with a `+' and a `-'. Then
\[ d_{ik} = \left\{
  \begin{array}{l l}
    +1 & \quad \text{if $i$ is the positive end of $k$}\\
    -1 & \quad \text{if $i$ is the negative end of $k$}\\
    0 & \quad \text{otherwise.}
  \end{array} \right.\]
Every node represents an aggregated area of generators and loads
and its dynamics are described by the so called `flux-decay' or `single-axis' model.
It extends the classical second order `swing equations' that describes the dynamics for the voltage angle $\delta$ and the frequency $\omega$ by including a differential equation describing voltage dynamics. A detailed derivation can be found e.g.~in \cite{machowski_power_2008}.
\smodc
The dynamics of node $i$ are given by:
 \be\ba{rcl}\label{revsystem1}
  \dot{\delta}_{i} &=&\omega^b_{i} \\
  M_{i}\dot{\omega}^b_{i} &=& u_{i} -  \sum_{j \in \mathcal{N}_{i}} V_{i}V_{j}B_{ij} \sin \big(\delta_{i} - \delta_{j}) \\&&- A_{i}\big(\omega^b_{i}-\omega^{n}\big) - P^l_i\\
   \frac{T_{doi}}{(X_{di}- X^{'}_{di})}\dot{V}_i  &=& \frac{E_{fi}}{(X_{di}- X^{'}_{di})} - \frac{1 -  B_{ii}(X_{di}- X^{'}_{di}) }{(X_{di}- X^{'}_{di})}V_i \\&&+ \sum_{j \in \mathcal{N}_{i}} V_{\smod j \emod} B_{ij} \cos \big(\delta_{i} - \delta_{j}).
  \ea\ee
where $B$ denotes the susceptance and $\mathcal{N}_i$ is the set of nodes connected to node $i$ by a transmission line.
 In high voltage transmission networks we consider here, the conductance is close to zero and therefore neglected, i.e. we assume that the network is lossless. An overview of the used symbols is provided in Table 1. \emodc
\begin{table}\label{tab1}
\setlength{\tabcolsep}{5pt}
\center
\begin{tabular} {l  l}
\hline\noalign{\smallskip}
&State variables  \\
\hline\noalign{\smallskip}
  $\delta_{i}$ & Voltage angle \\
  $\omega^b_{i}$ & Frequency \\
  \smodc $V_i$  & Voltage \emodc\\
  \hline\noalign{\smallskip}
  &Parameters \\
  \hline\noalign{\smallskip}
  $\omega^{n}$ & Nominal frequency, e.g. 50 or 60 Hz\\
  $M_{i}$ & Moment of inertia \\
  $A_{i}$ & Damping constant \\
  $\smod T_{doi} \emod$ & Direct axis transient open-circuit constant\\
  $X_{di}$ & Direct synchronous reactance \\
   $X_{di}^{'}$ & Direct synchronous transient reactance \\
   $B_{ij}$ & Susceptance of the transmission line\\
\hline\noalign{\smallskip}
&Controllable inputs \\
\hline\noalign{\smallskip}
  $u_{i}$ & Controllable power generation \\
  \smod $E_{fi}$  & Exciter voltage \emod \\
   \hline\noalign{\smallskip}
   &Uncontrollable inputs \\
\hline\noalign{\smallskip}
      $P^{l}_{i}$ & Power demand \\\\
\end{tabular}
\caption{Description of main variables and parameters appearing in the system model.}\end{table}
In this paper we focus on (optimal) frequency regulation and in order to keep the analysis concise we assume that $E_{fi}$ is constant and do not explicitly include exciter dynamics.
 To study the interconnected power network we write system (\ref{revsystem1}) compactly for all buses $i \in \mathcal{V}$ as
  \be\ba{rcl}\label{multi-machine4}
    \dot \eta &=& D^\trans \omega \\
  M\dot \omega &=& u - D \Gamma(V) \boldsymbol{\sin} (\eta) - A\omega - P^{l} \\
  T\dot{V} &=& -E(\eta)V + \overline E_{fd}\\
  y& = & \omega,
 \ea\ee
where $\omega$ is the frequency deviation $\omega^b - \omega^n$, $D$ is the incidence matrix corresponding to the topology of the network, $\Gamma(V)={\rm diag}\{\gamma_1,\ldots, \gamma_m\}$, with $\gamma_k=V_i V_j B_{ij} = V_j V_i B_{ji}$ and the index $k$ denoting the line $\{i,j\}$,  $\eta= D^T \delta$, $\overline E_{fd} = (\frac{E_{f1}}{(X_{d1}- X^{'}_{d1})}, \hdots, \frac{E_{fn}}{(X_{dn}- X^{'}_{dn})})^T$  and $E(\eta)$ is a matrix such that $E_{ii} =  \frac{1 \smod- \emod B_{ii}(X_{di}- X_{di}^{'}) }{X_{di}- X_{di}^{'}}$ and $E_{ij} = -B_{ij}\cos(\eta_k)$, where again the index $k$ denotes the line  $\{i,j\}$. We write explicitly the relation $y=\omega$, to stress that only the frequency is measured in the system (in contrast to e.g. \cite{li_connecting_2013}, where the controller design relies on power flow measurements as well).
\begin{rem}\label{remarkE}
      In a realistic network the reactance is higher than the transient reactance, i.e. $X_{di} > X^{'}_{di} > 0$
      and the self-susceptance $B_{ii}$ satisfies $B_{ii} \smod < \emod0$ and due to the shunt susceptance $\smod |B_{ii}| \emod > \sum_{j \in \mathcal{N}_{i}} |B_{ij}|$
      .
   It follows that $E(\eta)$ is a strictly diagonally dominant and symmetric matrix with positive elements on its diagonal and is therefore positive definite.
\end{rem}
\section{Incremental passivity of the multimachine power network}\label{sec.3}
\smodc The purpose of this section is  to show that system (\ref{multi-machine4}) is incrementally passive  when we consider $u$ as the input  and  $\omega$ as the output. This property turns out to be fundamental  in the subsequent analysis pursued in this paper. \emodc While showing the incremental passivity property, a storage function is derived, 
based
upon which the forthcoming analysis of  the response of system (\ref{multi-machine4}) to the power injection $\overline u$ and the  load $P^l$ is carried out.
Following  \cite{burger.depersis.aut13}, \cite{Burger2013b}, to show incremental passivity,
system (\ref{multi-machine4}) is first interpreted as two subsystems interconnected via constraints that reflect the topology of the network. As a matter of fact,
observe that  system (\ref{multi-machine4}) can be viewed as the feedback interconnection of the system
\be\label{multi-machine4a}
\ba{rcl}
M\dot \omega &=& u + \mu- A\omega - P^l  \\
y &=& \omega
\ea\ee
with the system
\be\label{multi-machine4b}
\ba{rcl}
\dot \eta &=& v \\
T\dot{V} &=& -E(\eta)V + \overline E_{fd} \\
\lambda &=& \Gamma(V) \boldsymbol{\sin} (\eta).
\ea\ee
These systems are interconnected  via the relations
\be\label{interc.rel}\ba{rcl}
v&=& D^\trans y\\
\mu &=& -D\lambda,
\ea\ee
where the incidence matrix $D$ reflects the topology of the network.
Before studying the incremental passivity of the system
it is convenient to recall its equilibria, which we will do in the next subsection.

\subsection{Equilibria}\label{subsec.eq}
As a first step we characterize the  constant steady state solution $(\overline \eta, \overline \omega, \overline V)$ of (\ref{multi-machine4}),  with a 
 generation $u=\overline u$, and in the case in which $\overline \omega$ is a constant belonging to the space $\mathcal{N}(D^T)$,  i.e.~it is a constant vector with all elements being equal. The steady state solution necessarily satisfies
\be\label{eq}\ba{rcl}
\mathbf{0} &=& D^\trans \overline \omega\\
\mathbf{0} &=&  \overline{u} - D \Gamma(\overline V) \boldsymbol{\sin} (\overline\eta)  - A\overline \omega - P^{l} \\
\mathbf{0} &=& -E(\overline \eta)\overline V + \overline E_{fd}.
\ea\ee
Notice that $\overline\eta$ is the vector of relative voltage angles that guarantee the power exchange among the buses at steady state.
The solution to (\ref{eq}) can be characterized as follows:
\begin{lem}\label{lemma1}
If there exists $(\overline \eta, \overline \omega, \overline V) \in \mathcal{R}(D^T)\times \R^n\times \R^n_{>0}$ such that (\ref{eq}) holds, then
necessarily $\overline \omega =\allones \omega_\ast$, with
\be\label{omega.star}
 \omega_\ast =   \dst\frac{ \allones^\trans (\overline{u}- P^{l})}{\allones^\trans A \allones}
= \dst\frac{\sum_{i \in \mathcal{V}} (\overline{u}_{i}-P^{l}_i)}{\sum_{i \in \mathcal{V}} A_i},
\ee
and the vector $\overline u-P^l$ must satisfy
\be\label{feasability}
\left(\dst I - \frac{A \allones \allones^\trans}{ \allones^\trans A \allones} \right)(\overline{u} -P^{l})\in \mathcal{D},
\ee
where
\be\ba{lcc}\label{set.D}
\mathcal{D}=\{v\in \mathcal{R}(D): \\
v= D
\Gamma(\overline V) \boldsymbol{\sin}(\overline \eta), \; \overline \eta\in \mathcal{R}(D^T), \overline V \in \mathbb{R}^n_{>0}\} .
\ea\ee
\end{lem}
Notice that, in view of (\ref{eq}),  the requirement for $\overline \omega$ to be a constant vector requires the vector $u-P^l$ to be constant as well.   \stopmodif
The proof of the lemma is straightforward and is therefore omitted. A characterization of the equilibria for a related system 
has been similarly discussed in \cite{simpson-porco_synchronization_2013}, \cite{schiffer_synchronization_????}, \cite{Zhao2014} and has its antecedents in e.g.~\cite{bergenTPAS81}.
Motivated by the result above, (\ref{feasability}) is introduced as a feasibility condition that formalizes the physical intuition that the network is capable of transferring the electrical power at its steady state.
\begin{assum}\label{assum2}
For a given $\overline u-P^l$, there exist $\overline \eta\in \mathcal{R}(D^T)$, $\overline V\in \R^n_{>0}$ and $\overline E_{fd} \in \R^{n}$ for which (\ref{feasability}) is satisfied and $\mathbf{0} = -E(\overline \eta)\overline V + \overline E_{fd}$.
\end{assum}
In some specific cases, the characterization above can be made more explicit.
If the graph has no cycles, then (\ref{feasability}) holds provided that $\overline u-P^l$ and $\overline V$
are such that (\cite{simpson-porco_synchronization_2013})
\[
\| \Gamma(\overline V)^{-1}D^\dagger\left(\dst I - \frac{A \allones \allones^\trans}{ \allones^\trans A \allones} \right)(\overline{u} -P^{l})\|_\infty <1,
\]
in which case $\overline \eta$ is obtained from
\[
 \boldsymbol{\sin}(\overline \eta)=\Gamma(\overline V)^{-1}D^\dagger\left(\dst I - \frac{A \allones \allones^\trans}{ \allones^\trans A \allones} \right)(\overline{u} -P^{l}),
\]
with $D^\dagger$ the Moore-Penrose pseudo-inverse.
\subsection{Incremental passivity of (\ref{multi-machine4})}
Having characterized the steady state solution of system (\ref{multi-machine4}) and having assumed that such a steady state solution exists, we are ready to state the main result of this section concerning the incremental passivity of the system with respect to the steady state solution.
The proof of the incremental passivity of system (\ref{multi-machine4}) can be split in a number of basic steps. First, one can show that system (\ref{multi-machine4a}) is incrementally passive with respect to the equilibrium solution, namely:
\begin{proposition}\label{prop1}
System
(\ref{multi-machine4a})   with inputs $u$ and
$\mu$ and output $y = \omega$,  is an output strictly incrementally passive system with
respect to a constant
 solution $\overline \omega$.
Namely,  there exists a regular storage function
$W_1(\omega, \overline\omega)$ which satisfies the incremental dissipation inequality
$
\dot W_1(\omega, \overline\omega)= -\rho (y-\overline y)+(y-\overline y)^\trans (\mu-\overline \mu)+(y-\overline y)^\trans (u-\overline u),
$
where $\dot W_1$ represents the directional derivative of $W_1$ along the solutions to  (\ref{multi-machine4a}) and $\rho:\R^n\to \R_{\ge 0}$ is a positive definite function.
\end{proposition}
\begin{proof}
Consider the regular storage function
$W_1(\omega, \overline\omega)=\frac{1}{2}(\omega-\overline\omega)^\trans M (\omega- \overline\omega).
$
We have
\[\ba{rcl}
\dot W_1&=&(\omega-\overline\omega)^\trans ( u + \mu - A \omega-P^l)\\
&=&(\omega-\overline\omega)^\trans (-A (\omega-\overline\omega) +(\mu-\overline \mu)+(u-\overline u))\\
&=&-(y-\overline y)^\trans A (y-\overline y)\\&&
+(y-\overline y)^\trans (\mu-\overline \mu)+(y-\overline y)^\trans (u-\overline u),
\ea\]
which proves the claim. Notice that in the second equality above, we have exploited the identity $\mathbf{0}= -A \overline \omega +\overline u- \overline \mu -P^l$.
\end{proof}
Second, we can prove a similar statement for system (\ref{multi-machine4b}) under the following condition:
 \begin{assum}\label{assum3}
   Let $\overline \eta \in (\frac{-\pi}{2}, \frac{\pi}{2})^m$ and $\overline V\in \R^n_{>0}$ be such that
\be\label{eq.assum3}\ba{rll}
 E(\overline\eta) &-& \smod{\rm diag}(\overline V)^{-1}\emod|D|\Gamma(\overline V)\text{diag}(\boldsymbol{\sin}(\overline\eta))\\&&
 \hspace{-0.5cm}\text{diag}(\boldsymbol{\cos}(\overline\eta))^{-1} \text{diag}(\boldsymbol{\sin}(\overline\eta)) |D|^T \smod{\rm diag}(\overline V)^{-1}\emod> 0,
  \ea\ee
%
  where $|D|$ is the incidence matrix with all elements positive.
 \end{assum}
The role of Assumption \ref{assum3} is to guarantee the existence of a suitable incremental storage function with respect to the constant solution $(\overline \eta, \overline V)$, as
becomes evident in the following lemma.

\begin{lem}\label{lemHessian}
  Let Assumption \ref{assum3} hold. Then the storage function
  \be \ba{rll}\label{w2} W_2(\eta, \overline\eta, V, \overline V) &=& -\mathbf{1}^T\Gamma(V)\boldsymbol{\cos}(\eta) + \mathbf{1}^T \Gamma(\overline V)\boldsymbol{\cos}(\overline \eta) \\&&- \smod \left(\Gamma(\overline V)\boldsymbol{\sin}(\overline \eta)\right)^T(\eta - \overline \eta) \\&&- \overline E_{fd}(V - \overline V)\\&& + \frac{1}{2}V^T F V - \frac{1}{2}\overline V^T F \overline V,
  \ea\ee where $F_{ii} =  \frac{1 \smod-\emod B_{ii}(X_{di}- X_{di}^{'}) }{X_{di}- X_{di}^{'}}$, has a strict local minimum at $(\overline \eta, \overline V)$.
\end{lem}
\begin{proof}
  First we consider the gradient of $W_2$, which is given by
  \be\ba{rll}
  \nabla W_2 &=& (\frac{\partial W_2}{\partial \eta} \frac{\partial W_2}{\partial V})^T \\
  &=& \begin{pmatrix}
    \Gamma(V)\boldsymbol{\sin}( \eta) - \Gamma (\overline V)\boldsymbol{\sin}(\overline \eta) \\
    E(\eta)V - \overline E_{fd}
  \end{pmatrix}.\nonumber
  \ea\ee
It is immediate to see that we have $\nabla W_2|_{\eta= \overline \eta, V=\overline V} = 0$. As the gradient of $W_2$ is zero at $(\overline \eta, \overline V)$,  for $W_2$ to have a strict local minimum it is sufficient that the Hessian is positive definite at $(\overline \eta, \overline V)$.
The Hessian is given by
  \be\ba{rll}
  \nabla ^2 W_2 = \begin{pmatrix}
    \Gamma(V)\text{diag}(\boldsymbol{\cos}(\eta))
    & \smod \mathcal{H}^T(\eta,V)\emod \\
    \smod \mathcal{H}(\eta,V)\emod  & E( \eta) \nonumber
  \end{pmatrix},
  \ea\ee
  where $ \smod \mathcal{H}(\eta,V) = {\rm diag}(V)^{-1}|D|\Gamma(V)\text{diag}(\boldsymbol{\sin}(\eta))$. \emod
Since $\Gamma(V)\text{diag}(\boldsymbol{\cos}({\eta}))$ is
positive definite for $\eta \in (\frac{-\pi}{2}, \frac{\pi}{2})^m$ it follows by invoking the Schur complement that $\nabla^2 W_2|_{\eta = \overline{\eta},V =\overline{V}} > 0$ if and only if
\be\ba{rll}
 E(\overline\eta) &-& \smod{\rm diag}(\overline V)^{-1}\emod|D|\Gamma(\overline V)\text{diag}(\boldsymbol{\sin}(\overline\eta))\\&&
 \hspace{-0.5cm}\text{diag}(\boldsymbol{\cos}(\overline\eta))^{-1} \text{diag}(\boldsymbol{\sin}(\overline\eta)) |D|^T \smod{\rm diag}(\overline V)^{-1}\emod> 0.
  \nonumber
  \ea\ee
  \end{proof}
 \begin{rem}
 Assuming $\overline \eta \in (\frac{-\pi}{2}, \frac{\pi}{2})^m$ is standard in power grid stability studies and is also referred to as a security constraint \cite{dorfler2014breaking}.
Assumption \ref{assum3} is a technical condition that allows us to infer boundedness of trajectories. An analogous condition (for a related model in a different reference frame) has been proposed in \cite{schiffer_synchronization_????}.     In the case of constant voltages Assumption \ref{assum3} becomes less restrictive and only the assumption $\overline \eta \in (\frac{-\pi}{2}, \frac{\pi}{2})^m$ is required (\cite{burger.et.al.mtns14b}). We notice indeed that by setting $V = \overline V$, the storage function (\ref{w2}) reduces to $-\mathbf{1}^T\Gamma(\overline V)\boldsymbol{\cos}(\eta) + \mathbf{1}^T \Gamma(\overline V)\boldsymbol{\cos}(\overline \eta) - \smod \left( \Gamma(\overline V)\boldsymbol{\sin}(\overline \eta)\right)^T(\eta - \overline \eta)$, which is regularly used in stability studies of the power grid (see e.g.~formula (22) in \cite{bergenTPAS81}, and also \cite{zhang.papa.acc13}) and has been adopted to study the stability of constant steady states of incrementally passive systems (\cite{burger.depersis.aut13}, \cite{Burger2013b}, \cite{Burger2013}).
\end{rem}
%
We are now ready to prove that feedback path (\ref{multi-machine4b}) is incrementally passive with respect to the equilibrium when Assumption \ref{assum3} holds.
\begin{proposition}\label{prop2}
Let \smod Assumptions \ref{assum2} and \ref{assum3} \emod hold. System
(\ref{multi-machine4b}) with  input $v$ and output $\lambda$  is an incrementally passive system, with respect to the constant equilibrium  $(\overline \eta$, $\overline V)$ which fulfills (\ref{eq.assum3}).  Namely, there exists a storage function
$W_2(\eta, \overline\eta, V, \overline V)$ which satisfies the incremental dissipation inequality
$
\dot W_2(\eta, \overline\eta, V, \overline V)= -\|\nabla_{V} W_2\|^2_{T^{-1}}+  (\lambda-\overline \lambda)^\trans (v-\overline v),
$
where $\dot W_2$ represents the directional derivative of $W_2$ along the solutions to  (\ref{multi-machine4b}) and  $\|\nabla_{V} W_2\|^2_{T^{-1}}$ is the shorthand notation for $(\nabla_{V} W_2)^TT^{-1}\nabla_{V} W_2$.
\end{proposition}


\begin{proof}
Consider the storage function $W_2$ given in (\ref{w2}). Under Assumption \ref{assum3} we have that  $W_2$ is a positive definite function in a neighborhood of $(\overline \eta, \overline V)$.
Since $T \dot{V}= -\nabla_{V} W_2$, it is straightforward to check that the dissipation inequality writes as
\[
\ba{rcl}
\dot W_{2}(\eta, \overline \eta, V, \overline V) &=& -\|\nabla_{V} W_2\|^2_{T^{-1}}  + (\Gamma(V)\boldsymbol{\sin}(\eta) \\&&-\Gamma(\overline V)\boldsymbol{\sin}(\overline \eta))^\trans\dot \eta\\
&=& -\|\nabla_{V} W_2\|^2_{T^{-1}}  +  (\lambda-\overline \lambda)^\trans (v-\overline v),\\
\ea\]
where the last equality trivially holds since $\dot{\overline \eta}=\overline v=\mathbf{0}$. This proves the claim.
\end{proof}
\vspace{1em}
The interconnection of incrementally passive systems via (\ref{interc.rel}) is known to be still incrementally passive. Bearing in mind Proposition \ref{prop1} and Proposition \ref{prop2} the next theorem follows immediately, proving that system (\ref{multi-machine4}) is output strictly incrementally passive with $u$ as an input and $y=\omega$ as an output.  We can exploit this feature to further design incrementally passive controllers that generate $u$ while establishing desired properties for the overall closed-loop system.
\begin{thm}\label{c1}
Let Assumptions
 \ref{assum2} \stopmodif and \ref{assum3} hold. System
(\ref{multi-machine4}) with input $u$ and output $y = \omega$ is an output strictly incrementally passive system, with respect to the constant equilibrium  $(\overline \eta, \overline \omega, \overline V)$ which fulfills (\ref{eq.assum3}).  Namely, there exists a storage function
$U(\omega, \overline\omega,\eta, \overline \eta, V, \overline V)=W_1(\omega, \overline\omega)+W_2(\eta, \overline \eta, V, \overline V)$ which satisfies the following incremental dissipation inequality
\be\ba{rll}
\dot U(\omega, \overline\omega,\eta, \overline \eta, V ,\overline V)&=& -\rho (y-\overline y) -\|\nabla_{V} W_2\|^2_{T^{-1}} \\ && + (y-\overline y)^\trans (u-\overline u), \nonumber
\ea\ee
where $\dot U$ represents the directional derivative of $U$ along the solutions to  (\ref{multi-machine4}) and $\rho$ is a positive definite function.
\end{thm}

\begin{proof}
The results descends immediately from Propositions \ref{prop1} and  \ref{prop2}  bearing in mind the interconnection constraints (\ref{interc.rel}).
\end{proof}

\begin{rem}
A function similar to $U$ (but in a different coordinate frame) was considered in e.g.~\cite{chu1999}. Here we provide a different construction that shows that $U$ is an incremental storage function with respect to which incremental passivity is proven. Highlighting this property is crucial in the approach and analysis we pursue. Furthermore, in the forthcoming analysis, we extend the storage function $U$ with a term that takes into account the addition of the controller and use it to infer convergence properties of the overall closed-loop system.
\end{rem}

The incremental passivity property of system (\ref{multi-machine4}) established above has  the immediate consequence that the response  of the system converges to an equilibrium when the power injection $\overline{u}$ and the load $P^l$ are
such that the total imbalance $\overline{u}-P^l$ is a constant.
For the sake of completeness, the details are provided in Corollary \ref{th1} below.
 We notice that an analogous study for related systems has been investigated in
e.g.~\cite[Section 7]{burger.depersis.aut13}, \cite{simpson-porco_synchronization_2013}, \cite{schiffer_synchronization_????} and \cite{Zhao2014}. 
Here,   similarly to \cite{schiffer_synchronization_????},  the study is carried out for  a third-order model with time-varying voltages and the result is an immediate consequence of the incremental passivity of the adopted model (\cite{burger.depersis.aut13} and \cite{Burger2013b}).
\begin{corollary}\label{th1}
 Let Assumptions 1 and 2 hold. There exists a neighborhood of initial conditions around the equilibrium $(\overline\eta, \overline \omega, \overline V)$,
such that the solutions to (\ref{multi-machine4}) starting from this neighborhood converge asymptotically to an equilibrium as characterized in Lemma \ref{lemma1}.
\end{corollary}

\begin{proof}
Bearing in mind Theorem \ref{c1} and setting  $u = \overline u$ and $y=\omega$,  the overall storage function $U(\omega, \overline\omega,\eta, \overline \eta, V, \overline V)=W_1(\omega, \overline\omega)+W_2(\eta, \overline \eta, V, \overline V)$ satisfies
\[\ba{rcl}
\dot U &=& -(\omega-\overline\omega)^\trans A (\omega-\overline\omega)
- (\omega-\overline \omega)^\trans \stopmodif D(\lambda-\overline \lambda) 
\\&& +(\lambda-\overline \lambda)^\trans D^\trans  (\omega-\overline \omega)\stopmodif -\|\nabla_{V} W_2\|^2_{T^{-1}} \\
&=& -(\omega-\overline\omega)^\trans A (\omega-\overline\omega) -\|\nabla_{V} W_2\|^2_{T^{-1}},
\ea\]
where we have exploited the fact that  $D^\trans \overline \omega=\mathbf{0}$, since $\overline \omega\in \mathcal{R}(\mbf{1})$. \stopmodif
As $\dot U\le 0$ and $(\overline\eta, \overline \omega, \overline V)$ is a strict local minimum as a consequence of Assumption \ref{assum3}, there exists a compact level set $\Upsilon$ around the equilibrium $(\overline\eta, \overline \omega, \overline V)$,
which is forward invariant.
 By LaSalle's invariance principle the solution starting in $\Upsilon$ asymptotically converges to the largest invariant set contained in
$\Upsilon\cap \{(\eta, \omega, V): \omega=\overline\omega, \|\nabla_{V} W_2\|=0\}$. \stopmodif Since we have $T \dot{V}= -\nabla_{V} W_2$,
 on such invariant set the system is
\[
\ba{rcl}
\dot \eta &=& \mathbf{0} \\\
\mathbf{0} &=& \overline u - A \overline \omega - D\Gamma(\tilde V) \boldsymbol{\sin} (\eta) - P^l \\
\mathbf{0} &=& -E(\eta)\tilde V+ \overline E_{fd},
\ea\]
where $\tilde V$ is a constant. From $\dot{\eta} = 0$ it follows that on the invariant set  $\eta$ is a constant $\eta = \tilde \eta$.
One can conclude that the system indeed converges to an equilibrium as  characterized in Lemma \ref{lemma1}.
\end{proof}

 \begin{rem}
   We cannot claim that $\tilde \eta = \overline \eta$ and $\tilde V = \overline V$, since the system could converge to any equilibrium within $\Upsilon$. This is due to the fact that we have not made any assumptions on the property of the equilibrium $(\overline \eta, \overline \omega, \overline V)$ being isolated. In order to establish that the equilibrium is isolated we should ask that the determinant of the Jacobian matrix at the equilibrium is nonsingular, as follows from the inverse function theorem.  This is not automatically guaranteed by $(\overline \eta, \overline \omega, \overline V)$ being a strict local minimum  of the storage function. To better elucidate this claim, first we notice that  system (\ref{multi-machine4}) can be written in the form
\be\ba{rll}\label{pH.multi-machine4}
\left(\ba{c}
\dot \eta\\ M \dot \omega \\ T \dot V
\ea\right)
=
\underbrace{\left[
\left(\ba{rrr}
\mathbf{0} & D^T & \mathbf{0} \\
-D & \mathbf{0} & \mathbf{0} \\
 \mathbf{0} & \mathbf{0} & \mathbf{0}
\ea\right)
-
\left(\ba{rrr}
\mathbf{0} & \mathbf{0} & \mathbf{0} \\
\mathbf{0} & A & \mathbf{0} \\
 \mathbf{0} & \mathbf{0} & I
\ea\right)
\right]}_{J-R} \\
\underbrace{\left(\ba{rrr}
I & \mathbf{0} & \mathbf{0} \\
\mathbf{0} & M^{-1} & \mathbf{0} \\
 \mathbf{0} & \mathbf{0} & I
\ea\right)}_{Q}\nabla U
+
\underbrace{\left(\ba{c}
\mathbf{0}\\ I \\ \mathbf{0}
\ea\right)}_{g}(u-\overline u), \nonumber
\ea\ee
 \end{rem}
  where $J$ is a skew-symmetric matrix and $R$ is a diagonal positive semi-definite matrix
 and
\[
\nabla U =
\begin{pmatrix}
\Gamma(V)\boldsymbol{\sin}( \eta) - \Gamma (\overline V)\boldsymbol{\sin}(\overline \eta) \\
M(\omega-\overline \omega)\\
E(\eta)V - \overline E_{fd}
    \end{pmatrix}.
\]
Set $u=\overline u$. Then LaSalle's invariance principle outlined in the proof above shows that
the solution converges to the largest invariant set where $\nabla U^T (J-R) \nabla U=0$, that is $\nabla U^T R \nabla U=0$. By the structure of $R$, the latter identity is equal to $\nabla_\omega U=\mathbf{0}$ (that is, $\omega =\overline \omega$) and  $\nabla_V U=\mathbf{0}$. In view of the second equation in (\ref{eq}) and of these identities,  on this largest invariant set we have $D(\Gamma(\tilde V)\boldsymbol{\sin}( \tilde \eta) - \Gamma (\overline V)\boldsymbol{\sin}(\overline \eta))=\mathbf{0}$. If $D$ has full-column rank, that is if the graph is acyclic, then $\Gamma(\tilde V)\boldsymbol{\sin}( \tilde \eta) - \Gamma (\overline V)\boldsymbol{\sin}(\overline \eta)=\mathbf{0}$. This would imply that any point on the invariant set satisfies $\nabla U = \mathbf{0}$ and it is therefore a critical point for $U$. Since we have assumed that $(\overline \eta, \overline \omega, \overline V)$ is a strict minimum for $U$ then we could conclude that every trajectory locally converges to $(\overline \eta, \overline \omega, \overline V)$. However, in the general case in which the graph is not acyclic, then there could be constant vector $(\tilde V, \tilde \eta)\ne \smod (\overline V, \overline \eta) \emod$ such that $D(\Gamma(\tilde V)\boldsymbol{\sin}( \tilde \eta) - \Gamma (\overline V)\boldsymbol{\sin}(\overline \eta))=\mathbf{0}$ (and $E(\tilde \eta)\tilde V - \overline E_{fd}=\mathbf{0}$). In this case, convergence can only be guaranteed to an equilibrium $(\tilde \eta, \overline \omega, \tilde V)$ characterized in Lemma  \ref{lemma1}, as remarked in the result above.
\section{Minimizing generation costs}\label{sec.4}
Before we address the design of controllers generating $u$, we discuss a desired optimality property the steady state input $\overline u$ should have. This is achieved by realizing that the share of total production each generator has to provide to balance the total electricity demand can be varied. Indeed, from equality (\ref{omega.star}) it can be seen that only the sum of the generators' production is important to characterize the steady state frequency. Generally, different generators have different associated cost functions, such that there is potential to reduce costs when the share of generation among the generators is coordinated in an economically efficient way (see also \cite{KTH}, \cite{dorfler2014breaking}, \cite{li_connecting_2013} and \cite{zhao2014decentralized}).
 In this section we characterize such an optimal generation that minimizes total costs. We consider only the costs of power generation $u$, as it is predominant over the \smodc excitation and transmission costs\emodc.
The corresponding network optimization problem we tackle is therefore as  follows:
\be\ba{rl}\label{optpr2}
&\dst\min_{u} C(u)=\dst \min_{u}\sum_{i \in \mathcal{V}} C_i(u_i)\\
{\rm s.t.} & 0= \allones ^T (u - P^l),
\ea\ee
where $C_i(u_i)$ is a strictly convex cost function associated to generator $i$.  Comparing the equality constraint to (\ref{omega.star}), it is immediate to see that the solution to (\ref{optpr2}) implies a zero frequency deviation at steady state. The relation of (\ref{optpr2}) with the zero steady state frequency deviation as characterized in (\ref{eq}) with $\overline \omega = 0$ will be made more explicit at the end of this section.
Following standard literature on convex optimization 
we introduce the Lagrangian function
$
L(u, \lambda)
= C(u)+\lambda \allones^T \left(
u - P^l \nonumber
\right)$,
 where $\lambda \in \R$ is the Lagrange multiplier.
Since $C(u)$ is strictly convex we have that $L(u, \lambda)$ is strictly convex in $u$
and
  concave in $\lambda$. Therefore there exists a saddle point solution to
$
\max_{\lambda}\min_{u} L(u,\lambda).
$
Applying first order optimality conditions, the saddle point $(\overline u, \overline \lambda)$ must satisfy
\be\label{eq.saddle.point}\ba{rll}
\nabla C(\overline u)  +\allones \overline \lambda &=&\mathbf{0}\\
\allones^T (\overline u
 - P^l)&=& 0.
\ea\ee
In the remainder we assume that $C(u)$ is quadratic\footnote{\startmodif The results in this work holds for linear-quadratic cost functions as well, i.e. $C(u) = \frac{1}{2}u^T Q u + R^Tu + \allones ^T S$. In this case $\overline u = Q^{-1}(\overline \theta - R)$, where $\overline \theta = \frac{\allones \allones^T (P^l + Q^{-1}R)}{\allones^T Q^{-1}\allones} \in \mathcal{R}(\allones)$. For the sake of brevity we focus in this work on the quadratic case.
\stopmodif},
i.e. \startmodif $C(u) = \frac{1}{2}u^TQu = \sum_{i \in \mathcal{N}} \frac{1}{2}q_iu_i^2$, \stopmodif with $q_i>0$.
We make now explicit the solution to the previous set of equations in the case of quadratic cost functions.

\begin{lem}\label{l3}
Let $C(u)=\frac{1}{2} {u}^\trans Q u$, with $Q>0$ and diagonal. There exists a solution $(\overline u, \overline \lambda)$ to (\ref{eq.saddle.point}) if and only if
the optimal control is
\be\label{optimal.u}
\overline u = Q^{-1} \dst\frac{\allones\allones^\trans P^l}{\allones^\trans Q^{-1}\allones} ,
\ee
and the optimal Lagrange multiplier is
\[
\overline\lambda = \left(
-\dst\frac{\allones^\trans P^l}{\allones^\trans Q^{-1}\allones}
\right).
\]
\end{lem}

The proof is standard and is omitted. For the optimal control characterized above to guarantee a zero frequency deviation, the equalities (\ref{eq}) should now be satisfied with $\overline u$ as in (\ref{optimal.u}) and $\overline \omega =\mathbf{0}$. In this case, the second equality becomes
\be\label{optimal.eta}
D \Gamma(\overline V) \boldsymbol{\sin} (\overline \eta) = (Q^{-1}\frac{ \allones \allones^T}{\allones^T Q^{-1} \allones} -I_n )P^l .
\ee
The equality (\ref{optimal.eta}) shows that an optimal solution may require a nonzero $D \Gamma(\overline V) \boldsymbol{\sin} (\overline \eta)$ at steady state. That implies that at steady state power flows may be exchanged among the buses in the network and that the local demand $P^l_i$ may not necessarily be all compensated by  $\overline u_i$. In fact, from (\ref{optimal.u}) it is seen that to balance the overall demand $\mathbf{1}^T P^l$ each generator should contribute an amount of power  that is inversely proportional to its marginal cost $q_i$.   From (\ref{optimal.u}), we also notice that  the optimal power generation is independent of the steady state voltage $\overline V$.
 Motivated by Lemma \ref{l3} and the remark that led to (\ref{optimal.eta}),  we introduce the following condition that replaces the previous Assumption \ref{assum2}:
\begin{assum}\label{assum4}
For a given  $P^l$, there exist $\overline \eta\in \mathcal{R}(D^T)$,    $\overline V\in \R^n_{>0}$   and $\overline E_{fd}\in \R^n$  for which
\be\label{Pl.optimal.cond}
(Q^{-1}\frac{ \allones \allones^T}{\allones^T Q^{-1} \allones} - I_n )P^l \in \mathcal{D},
\ee
 with $\mathcal{D}$ defined as in Lemma \ref{lemma1},
is satisfied and $\mathbf{0} = -E(\overline \eta)\overline V + \overline E_{fd}$.
\end{assum}
We can relate optimization problem (\ref{optpr2}) to  another optimization problem in which the zero frequency deviation requirement at steady state is more explicit.\footnote{The authors thank Nima Monshizadeh for suggesting this lemma.}
\begin{lem}\label{lemeq}
  Let Assumption \ref{assum4} hold and let $C(u)=\frac{1}{2} {u}^\trans Q u$, with $Q>0$ and diagonal. Then the optimal $\overline u$ solving (\ref{optpr2}) is equivalent to the optimal $\overline u'$ solving
  \be\ba{rl}\label{optpr1}
&\dst\min_{u, \eta} C(u)=\dst \min_{u, \eta}\sum_{i \in \mathcal{V}} C_i(u_i)\\
{\rm s.t.} & \boldsymbol{0}= u -
D\Gamma(\overline V) \boldsymbol{\sin}(\eta) - P^l \\
& \eta\in \mathcal{R}(D^T).
\ea\ee
\end{lem}
\begin{proof}
  By multiplying both sides of the equality constraint of (\ref{optpr1}) from the left by $\allones^T$, we obtain the constraint of (\ref{optpr2}). Hence, $\overline u'$ satisfies (\ref{optpr2}), and we have $C(\overline u) \leq C(\overline u')$. By the equality constraint in (\ref{optpr2}), we have $\overline u-P^l \in \mathcal{R}(\allones)^\bot$. Thus, $\overline u-P^l \in \mathcal{N}(D^T)^\bot$ which yields $\overline u-P^l \in \mathcal{R}(D)$. Therefore, $\overline u-P^l = Dv$ for some vector $v$. By the choice $v=\Gamma(\overline V) \boldsymbol{\sin}(\overline \eta)$, which exists under Assumption $\ref{assum4}$, $\overline u-P^l = Dv$ satisfies (\ref{optpr1}) and we have $C(\overline u') \leq C(\overline u)$. Consequently, $C(\overline u') = C(\overline u)$ which results in $\overline u' = \overline u$ due to the strict convexity of $C$.
\end{proof}
\vspace{1em}
Lemma \ref{lemeq} provides insights on how the nonconvex optimization problem (\ref{optpr1}) can be solved for $\overline u'$ by (\ref{optpr2}) without  the approximation $\boldsymbol{\sin}(\eta) = \eta$ as long as Assumption \ref{assum4} holds.  This can be seen as an alternative approach to  solving for (\ref{optpr1}) by an equivalant `DC' problem (see e.g. \cite{dorfler2014breaking}) where the constraint reads as $\boldsymbol{0}= u -
D\Gamma(\overline V) \eta^{DC}$ and requires the graph to be a tree \cite{dorflernovel}.  The characterization of $\overline u$ in (\ref{optimal.u}) will enable the design of controllers regulating the frequency in an optimal manner, which we pursue in the next section. We also remark that even in the case in which $P^l$ is a time-varying signal, the optimal power generation control that guarantees a zero frequency deviation is still given by $\overline u$ in (\ref{optimal.u}). This property will be used in Section \ref{sec.6}. Finally, we notice that, following \cite{burger.depersis.aut13}, the optimal generation $\overline u$ characterized above can be interpreted as the optimal feedfoward control which solves the regulator equations connected with the frequency regulation problem. We will elaborate on this more in the next section.

\smod
\begin{rem}
It is worth stressing that the explicit request of having $\overline V\in \R^n_{> 0}$ in Assumption  \ref{assum2}, \ref{assum3} and \ref{assum4} is not necessary\footnote{The authors thank Hong-keun Kim for this remark.}. As a matter of fact, for any $\overline V$ which satisfies $-E(\overline\eta) \overline V +\overline E_{fd}=\mathbf{0}$, it trivially holds true that $\overline V =E(\overline\eta)^{-1}\overline E_{fd}$. Let $\overline E_{fd}\in \R^n_{> 0}$. Since  $E(\overline\eta)$ has all the off-diagonal entries non-positive, then it is inverse-positive  (\cite{Plemmons1977175}, Theorem 1, ${\rm F_{15}}$), i.e.~each entry of the inverse $E(\overline\eta)^{-1}$ is non-negative.
%
%
Furthermore, since $E(\overline\eta)^{-1}$ is invertible, each row has at least one strictly positive entry. Therefore, the product $\overline V =E(\overline\eta)^{-1}\overline E_{fd}$ must necessarily return a vector with all strictly positive entries.
\end{rem}
\emod
\section{Economically efficient frequency regulation in the presence of constant power demand }\label{sec.contr}
Corollary \ref{th1} shows attractivity of the steady state solution under a constant imbalance vector $u-P^l$, which generally results in a nonzero steady state frequency deviation. In this section we consider the problem of designing the generation $u$ in such a way that at steady state the system achieves a zero frequency deviation. We adopt the framework provided in \cite{burger.et.al.mtns14b}, \cite{burger.depersis.aut13}, \cite{Burger2013b}. This framework provides a constructive and straightforward procedure to the design of the frequency regulator.

We start the analysis by reminding that Theorem \ref{c1} states the  incremental passivity property of the system
\be\label{multi-machine-overall2}
\ba{rcl}
\dot \eta &=& D^\trans \omega \\
M\dot \omega &=& u - A \omega  -D 
\Gamma \boldsymbol{\sin} (\eta) - P^l \\
T\dot{V} &=& -E(\eta)V + \overline E_{fd} \\
y &=& \omega.
\ea\ee
The incremental passivity property  holds with respect to two solutions  of (\ref{multi-machine-overall2}). As one of the two solutions, we adopt here a
solution to the {\em regulator equations} (\ref{re-multi-machine-overall2}) below (\cite{burger.depersis.aut13}, \cite{Burger2013b}).  This is  the  state 
$(\overline \eta,\overline \omega, \overline V)$, the feedforward input $\overline u$ and the output $\overline y=\overline \omega=\mathbf{0}$
such that
\be\label{re-multi-machine-overall2}
\ba{rll}
\dot {\overline \eta} &=& D^\trans \overline \omega \hspace{0.5em} =  \hspace{0.5em} \mathbf{0} \\
\mathbf{0} &=& \overline  u -D \Gamma(\overline V) \boldsymbol{\sin} (\overline \eta) - P^l \\
\mathbf{0} &=& -E(\overline \eta)\overline V + \overline E_{fd} \\
\overline y &=& \overline \omega  \hspace{0.5em} =  \hspace{0.5em} \mathbf{0}.
\ea\ee
Among the many possible choices, we focus on  the steady state  solution that arises from the solution of the optimal control problem in the previous section, namely
\be\label{optimal.uX}
\overline u =
Q^{-1} \dst\frac{\allones\allones^\trans P^l}{\allones^\trans Q^{-1}\allones} ,
\ee
characterized in (\ref{optimal.u}) above, and $\overline \eta$ such that
\be\label{optimal.eta2}
D \Gamma(\overline V) \boldsymbol{\sin} (\overline \eta) = (Q^{-1}\frac{ \allones \allones^T}{\allones^T Q^{-1} \allones} - I_n)P^l . \nonumber
\ee
The framework presented in \cite{burger.depersis.aut13} and \cite{Burger2013b} prescribes to design an incrementally passive feedback controller that is able to generate the feedforward input (\ref{optimal.uX}). The interconnection of the process (\ref{multi-machine-overall2}) and of the incrementally feedback controller to be introduced below yields
 a closed-loop
system whose solutions asymptotically converge to the desired steady state solution. This idea is made precise in the \startmodif theorem \stopmodif below, that is
the main result of the section and where we propose a dynamic controller that converges asymptotically to the optimal feedforward input that guarantees zero frequency deviation. The result deals with constant power demand, the extension to time-varying power demands being postponed to a later section. \startmodif As will become clear it is essential that the controllers exchange information, leading to the following assumption.
\begin{assum}\label{assum5}
  The undirected graph reflecting the topology of information exchange among the nodes is connected.
\end{assum}
\begin{thm}\label{th5}
Consider the system
(\ref{multi-machine-overall2}) with constant power demand $P^l$ and let Assumptions  \smod \ref{assum3}, \ref{assum4} and \ref{assum5} \emod hold.
 Then controllers at the nodes\footnote{\smod For linear-quadratic cost functions the controller output becomes $u_i = q_i^{-1}(\theta_i - r_i)$. \emod}
\be\label{distr.contr}\ba{rcl}
\dot \theta_i &=& \sum_{j\in \mathcal{N}_i^{comm}} (\theta_j-\theta_i) - q_i^{-1}\omega_i \\
u_i &=& q_i^{-1} \theta_i,
\ea\ee
for $i=1,2,\ldots, n,$ where $\mathcal{N}^{comm}_i$
denotes the set of neighbors of node $i$ in a graph describing the exchange of information among the controllers,
guarantee the solutions to the closed-loop system that start in a neighborhood of  $( \overline \eta, \overline \omega, \overline V, \overline \theta)$  to converge asymptotically to the largest invariant set where $\omega_i=0$ for all $i=1,2,\ldots, n$, $\|\nabla_V W_2\|=0$ (thus, $V=\tilde V$ is a constant), and $\theta=\overline\theta$,
$\overline\theta$ being the vector
\[
{\overline\theta} = \dst\frac{\allones\allones^\trans P^l}{\allones^\trans Q^{-1}\allones},
\]
such that
$
\overline u = Q^{-1} \overline\theta
$
satisfies
\[\ba{rcl}
\dot{\tilde\eta} &=& \mathbf{0}\\
\mathbf{0} &=& \overline u - D\Gamma  (\tilde V)  \boldsymbol{\sin}(\tilde \eta) - P^l\\
\mathbf{0} &=& -E(\tilde \eta)\tilde V + \overline E_{fd}\\
y &=& \mathbf{0}.
\ea\]
\end{thm}
\begin{proof}
Bearing in mind Theorem \ref{c1}, one can notice  that the incremental storage function
$
U(\omega, \overline\omega, \eta, \overline\eta, V, \overline V)=W_1(\omega, \overline\omega)+W_2(\eta, \overline\eta, V, \overline V)
%
$
satisfies
$
\dot U = -(\omega-\overline\omega)^\trans A (\omega- \overline\omega) -\|\nabla_{V} W_2\|^2_{T^{-1}} +
(\omega-\overline\omega)^\trans (u-\overline u),
$
thus showing that the system is output strictly incrementally passive. This equality holds in particular for $\overline\omega=\mathbf{0}$, $\overline u$ given in (\ref{optimal.u}) and $\overline \eta, \overline V$ as in Assumption \ref{assum4}.
The internal model principle design pursued in \cite{burger.depersis.aut13} and \cite{Burger2013b} prescribes the design of a controller able to generate the feedforward input   $\overline u$. To this purpose,
we introduce the overall controller
%
%
%
%
\be\label{stacked.controller}\ba{rcl}
\dot \theta &=& -L_{comm} \theta +\overline H^T v\\
u &=& \overline{H} \theta,
\ea\ee
where $\theta\in \R^n$,
$L_{comm}$ the Laplacian associated with a graph that describes the exchange of information among the controllers, and with the term
$
\overline H^T v
$
needed to guarantee the incremental passivity property of the controller (see \cite{burger.depersis.aut13}, \cite{Burger2013b} for details). Here $v\in \R^n$ is an extra control input to be designed later, while $\overline H=\overline H^T = Q^{-1}$.
If $v=\mathbf{0}$ and $\overline \theta(0)= \dst\frac{\allones\allones^\trans P^l}{\allones^\trans Q^{-1}\allones} $,
then $\overline \theta(t):=\overline \theta(0)$ satisfies the differential equation in (\ref{stacked.controller}) and moreover the corresponding output $ \overline{H}\, \overline \theta(t)$ is identically equal to the feedforward input $\overline u(t)$ defined in (\ref{optimal.u}), provided that $\overline H= Q^{-1}$. More explicitly, we have
\be\label{stacked.controller.w}\ba{rcl}
\dot{\overline\theta} &=& -L_{comm}\overline\theta\\
\overline u &=& \overline H \,\overline \theta.
\ea\ee
 Notice that this is a manifestation of the internal model principle (\cite{burger.depersis.aut13}, \cite{Burger2013b}), that is the ability of the controller to generate, in open-loop and when properly initialized, the prescribed feedforward input.
Consider now the incremental storage function
$
\Theta(\theta,\overline \theta)=\frac{1}{2} (\theta-\overline \theta)^T (\theta-\overline \theta).
$
It satisfies
\[\ba{rcl}
\dot \Theta(\theta,\overline \theta)&=&(\theta-\overline \theta)^T (-L_{comm} \theta +\overline H^T v +L_{comm}\overline \theta)\\
&= & -(\theta-\overline \theta)^T L_{comm}(\theta-\overline \theta)+ (\theta-\overline \theta)^T\overline H^T v \\
&=&-(\theta-\overline \theta)^T L_{comm}(\theta-\overline \theta) +(u-\overline u)^T v.
\ea\]
We now interconnect the 
third-order model (\ref{multi-machine-overall2}) and the controller
(\ref{stacked.controller}), obtaining
\[
\ba{rcl}
\dot \eta &=& D^\trans \omega \\
 M\dot \omega &=& \overline{H} \theta  -D \Gamma(V) \boldsymbol{\sin} (\eta) - A \omega  - P^l \\
 T\dot{V} &=& -E(\eta)V + \overline E_{fd} \\
\dot \theta &=& -L_{comm} \theta +\overline H^T v\\
y &=& \omega.
\ea
\]
Observe that the quadruple $(\overline\eta, \overline\omega, \overline V, \overline\theta)$ is a  solution to the closed-loop system just defined when $v=\mathbf{0}$.
Consider  the incremental storage function
\[
Z(\eta, \overline\eta, \omega, \overline\omega, V, \overline V,\theta,\overline\theta)=U(\eta, \overline\eta,\omega, \overline\omega, V, \overline V)+\Theta(\theta,\overline\theta),
\]
where $( \overline \eta, \overline V)$ fulfills Assumption \ref{assum3}.
Following the arguments of Lemma \ref{lemHessian}, it is immediate to see that under condition  (\ref{eq.assum3}) 
we have that $\nabla Z|_{\eta= \overline \eta, \omega=\overline \omega, V=\overline V, \theta = \overline \theta} = 0$ and
 $\nabla^2 Z|_{\eta= \overline \eta, \omega=\overline \omega, V=\overline V, \theta = \overline \theta} > 0$,
such that $Z$ has a strict local minimum at  $( \overline \eta, \overline \omega, \overline V, \overline \theta)$. It turns out that
\[\ba{rll}
\dot Z &=& -(\omega-\overline\omega)^\trans A (\omega- \overline\omega)  -\|\nabla_{V} W_2\|^2_{T^{-1}} \\&& +
(\omega-\overline\omega)^\trans (u-\overline u)
-(\theta-\overline\theta)^T L_{comm} (\theta-\overline\theta)\\&&+(u-\overline u)^T v.
\ea\]
As we are still free to design $v$, the choice
$
v= -(\omega-\overline\omega)= -\omega
$
returns
\be\ba{rll}
\dot Z&=& -(\omega-\overline\omega)^\trans A (\omega- \overline\omega)  -\|\nabla_{V} W_2\|^2_{T^{-1}} \\&& -(\theta-\overline\theta)^T L_{comm} (\theta-\overline\theta) \leq 0, \nonumber
\ea\ee
thus showing that $Z$ is bounded.
 As $\dot{Z} \leq 0$,
 there exists a compact level set $\Upsilon$
around the equilibrium $( \overline \eta, \overline \omega, \overline V, \overline \theta)$
which is forward invariant.
By LaSalle's invariance principle the solution starting in $\Upsilon$ asymptotically converges to the largest invariant set contained in $\Upsilon \cap \{
( \overline \eta, \overline \omega, \overline V, \overline \theta):  \overline \omega=\mathbf{0}, \|\nabla_{V} W_2\|=0, \theta = \overline \theta + \mathbf{1}_n \alpha\}$, where $\alpha:\R_{\ge 0}\to \R$ is a function \smod and $\theta = \overline \theta + \mathbf{1}_n \alpha$ follows from the communication graph being connected, i.e. $ \mathcal{N}(L_{comm}) = \mathcal{R}(\mathbf{1}_n)$. \emod On such invariant set the system is
  \be\ba{rll}\label{invariant_freq}
\dot{\eta} &=& D^T \overline \omega \hspace{0.5em}=\hspace{0.5em}\mathbf{0} \\
\boldsymbol{0} &=&  -D (\Gamma(\tilde V) \boldsymbol{\sin}(\eta) - \Gamma(\overline V) \boldsymbol{\sin}(\overline \eta)) -\overline H\mathbf{1}_n\alpha
\\ \boldsymbol{0} &=& -E(\eta)\tilde V + \overline E_{fd}  \\
\dot {\overline \theta} +\mathbf{1}_n\dot \alpha&=& -L_{comm} (\overline \theta +\mathbf{1}_n \alpha),
\ea\ee
where $\tilde V$ is a constant. From $\dot{\eta} = 0$ it follows that on the invariant set  $\eta$ is a constant $\eta = \tilde \eta$. In the second equality above, we have exploited the identity $\mathbf{0}= \overline H\overline \theta - D\Gamma(\overline V) \boldsymbol{\sin}(\overline \eta) -A \overline \omega -P^l$.
Bearing in mind that $\overline H=Q^{-1}$, it follows that necessarily $\alpha=0$ (it is sufficient to multiply both sides of the second line in (\ref{invariant_freq}) by $\allones^T$). Hence on the invariant set $\theta=\overline\theta$ and the output of the controller is $\overline H \overline\theta$ which equals the optimal feedforward input (\ref{optimal.uX}). We conclude that the dynamical controller guarantees asymptotic regulation to zero of the frequency deviation and convergence to the optimal feedforward input.
\end{proof}
%
%
%
%
%
%
\vspace{1em}
The interpretation of the theorem is straightforward: it shows that the dynamic  controllers based on an internal model design synchronize to a steady state solution of the exosystem that generates the feedforward input  that minimizes generation costs and is able to guarantee a zero frequency deviation. These controllers must be initialized in the vicinity of $\overline \theta$ which represents a nominal  estimate of the total demand. Starting from this initial guess, the controllers adjust the power production depending on the frequency deviation which in turn depends on {\em actual} (and {\em unmeasured}) demand.
\begin{rem}
The use of an auxiliary communication graph to allow the exchange of information among the controllers at the node\smod s \emod is commonly found in wide area control of the power grid and has been suggested in \cite{simpson-porco_synchronization_2013} for control of microgrids as well.
  The Laplacian matrix $L_{comm}$ reflecting the exchange of information among the nodes is introduced here in order to prove convergence of $u$ to the optimal $\overline u$. Note that the communication graph can differ from the graph describing the transmission network \cite{simpson-porco_synchronization_2013} \smod as long as the communication graph is connected\emod. The distributed nature of the controllers can be relaxed resulting in fully decentralized controllers without exchange of information at the price that optimality can not be guaranteed anymore. \smod The robustness of such decentralized controllers to e.g. noise and delays has attracted a considerable amount of attention (\cite{Overview}, \cite{KTH}, and references therein) and its practical applicability needs to be evaluated carefully.  \emod
\end{rem}
\section{Frequency regulation in the presence of time-varying power demand}\label{sec.6}
Until now we assumed that the power demand term $P^l$ is unknown but constant, as is a standard practice in current research. Future smart grids should however be able to cope with rapid fluctuations of the power demand at the same timescale as the dynamics describing the physical infrastructure, such that approximating the power demand by a constant can become unrealistic.  This asks for controllers able to deal with time-varying power demand.  In the previous section we studied within the framework of \cite{burger.depersis.aut13} dynamical controllers able to achieve zero frequency deviation with steady state optimal production in the presence of constant power demand. Since the framework of \cite{burger.depersis.aut13} lends itself to deal with time-varying disturbances, it is natural to wonder whether the approach can be used to design frequency regulators in the presence of time-varying power demand. This is investigated in this section.\\
  Although the power demand is not known, we will assume that it is the output of a known exosystem, as it is customary in output regulation theory. Let $P^l$ depend linearly on $\sigma$, namely, let
\be\ba{rll}\label{exos2}
P^l=\Pi \sigma,
\ea\ee
for some matrix $\Pi$, where $\sigma$ is
the state variable of the exosystem
\be\label{exos1}
\dot \sigma=s(\sigma).
\ee
Here the map $s$ is assumed to satisfy the incremental passivity property $(s(\sigma)-s(\sigma'))^T (\sigma-\sigma')\le 0$ for all $\sigma,\sigma'$. It will be useful to limit ourselves to  the case $s(\sigma)=S\sigma$, with $S$ a skew-symmetric matrix. In this case, the exosystem (\ref{exos2}), (\ref{exos1}) generates linear combinations of constant and sinusoidal signals. We will however continue to refer to $s(\sigma)$ for the sake of generality, using explicitly $S\sigma$ only when needed. \smodc The choice (\ref{exos1}) is further motivated by spectral decomposition of load patterns \cite{Aguirre200873}, ocean wave energy \cite{Falnes2007185} and wind energy \cite{van1957power}, \cite{milan2013turbulent} that indicate that the net load can indeed be approximated by a superposition of a constant and a few sinusoidal signals. \emodc More explicit, we model the power demand $P^l_i$ as a superposition of a constant power demand
($\Pi_{1i}\sigma_{1}$),
a periodic power demand that can be compensated optimally ($\Pi_{2i}\sigma_{2}(t)$) and a periodic power demand that cannot be compensated optimally ($\Pi_{3i}\sigma_{3i}(t)$), such that
$P^l_i(t) = 
\Pi_{1i}\sigma_{1}  + \Pi_{2i}\sigma_{2}(t) + \Pi_{3i}\sigma_{3i}(t)$. The reason why we distinguish between $\Pi_{2i}\sigma_{2}$ and $\Pi_{3i}\sigma_{3i}$ becomes evident in the next subsection.
Similarly we write the steady state input as a sum of its components, $\overline u_i(t) = \overline u_{1i} + \overline u_{2i}(t) + \overline u_{3i}(t)$. The explicit dependency on time will be dropped in the remainder and was added here to stress the differences between constant and time-varying signals.
\begin{example}\label{example}
Consider the case of a periodic power demand with frequency $\mu$ superimposed to a constant power demand. This demand can be modeled as
$
P^l_i= \Pi_{1i} \sigma_1 + \Pi_{2i} \sigma_2
$
where $\dot \sigma =S \sigma$ with
\[
S=\left(\ba{rr}
0 & \mathbf{0}_{1\times 2} \\
\mathbf{0}_{2\times 1} & S_2
\ea\right)=\left(\ba{rrr}
0 & 0 & 0 \\
0 & 0 & \mu\\
0 & -\mu & 0 \\
\ea\right),
\]
$\Pi_{1i}$ is a real number and $\Pi_{2i}=q_i^{-1} (1\;0) = q_i^{-1}R_2 $. In this case $R_2=(1\;\;0)$ and notice that the pair $(R_2, S_2)$ is observable.
\end{example}
\subsection{Economically efficient frequency regulation in the presence of a class of time-varying power demand}
We focus in this subsection on the case in which the power demand at each node has the form $P^l_i = \Pi_{1i}\sigma_{1} + \Pi_{2i}\sigma_{2}$, where $\sigma_1, \sigma_2$ will be specified below. At steady state we have that $\dot{\eta} =\mathbf{0}$, $\dot{V} = \mathbf{0}$ and therefore power flows between different control areas need to be constant. This observation restricts the class of time-varying power demand that can be compensated for by an optimal generation $u$. We will make this more specific.
Recall that the optimal power generation at steady state is given by
\be\label{optimal.u2}
\overline u =
Q^{-1} \dst\frac{\allones\allones^\trans P^l}{\allones^\trans Q^{-1}\allones} ,
\ee
characterized in (\ref{optimal.u}) above. In this case, the second equality in  (\ref{re-multi-machine-overall2}) writes as  in (\ref{optimal.eta})
\be\label{optimal.eta2}
D \Gamma(\overline V) \boldsymbol{\sin} (\overline \eta) = (Q^{-1}\frac{ \allones \allones^T}{\allones^T Q^{-1} \allones} -I_n)P^l .
\ee
This implies that the quantity on the right-hand side must be  constant and that there must exist  a vector $\overline \eta\in \mathcal{R}(D^T)$ which satisfies the equality.
If we differentiate in the disturbance term $\Pi \sigma$ between a constant component $\Pi_1 \sigma_1$ and a time-varying component $\Pi_2 \sigma_2$, i.e. $\Pi \sigma=\Pi_1 \sigma_1+\Pi_2 \sigma_2$, and there exists a solution to the identity (\ref{optimal.eta2}) when $\Pi \sigma$ is replaced by $\Pi_1 \sigma_1$, then such a solution continues to exist provided that  the time-varying component of $\Pi \sigma$ belongs to the null space of $(Q^{-1}\frac{ \allones \allones^T}{\allones^T Q^{-1} \allones} -I_n)$. The null space above can be easily characterized.
\begin{lem}\label{lemnullspace}
  The null space of $(Q^{-1}\frac{ \allones \allones^T}{\allones^T Q^{-1} \allones} -I_n)$ is given by $\mathcal{R}(Q^{-1}\allones)$.
\end{lem}
\begin{proof}
First consider the matrix
$
{-\allones^T Q^{-1} \allones}\cdot (Q^{-1}\frac{ \allones \allones^T}{\allones^T Q^{-1} \allones}- I_n),
$
which takes the expression
\[
L^T= \left(\ba{rrrr}
L_{11}^T & -q_1^{-1} & \ldots & -q_1^{-1}\\
-q_2^{-1} & L_{22}^T & \ldots & -q_2^{-1}\\
\vdots &\vdots &\vdots &\vdots \\
-q_n^{-1} & -q_n^{-1} & \ldots & L_{nn}^T\\
\ea\right),
\]
where
$L_{ii}^T= (\sum_{j \in \mathcal{V} \backslash \{{i}\}} q_j^{-1})$. Hence, $L$ is the Laplacian matrix of a weighted complete graph. The rank of the Laplacian matrix of a connected graph is $n-1$. Thus the rank of the
matrix $L^T$ is also $n-1$. Since the rank of a matrix is not altered by the multiplication by a nonzero constant, one infers that the matrix $(Q^{-1}\frac{ \allones \allones^T}{\allones^T Q^{-1} \allones} - I_n)$ has rank $n-1$ as well.
Thus its null space has dimension $1$. Now, it is easily checked
that the range of $Q^{-1}\allones$ is included in the null space of $(Q^{-1}\frac{ \allones \allones^T}{\allones^T Q^{-1} \allones} - I_n)$.
\end{proof}
\vspace{1em}
From Lemma \ref{lemnullspace} it follows that the time varying component $\Pi_2 \sigma_2$ of the unknown demand must satisfy
$
\Pi_2 \sigma_2 \in \mathcal{R}(Q^{-1}\allones).
$
This leads to the following model for the power demand
\be\label{exosystem.revised}
\ba{rll}
\dot \sigma_1 &=&0 \\
\dot \sigma_2 &=& s_2(\sigma_2)\\
P^l &=& \Pi_1 \sigma_1 +Q^{-1}\allones R_2 \sigma_2,
\ea\ee
where $\Pi_1$ is a diagonal matrix, $R_2$ is some suitable row vector such that the pair $(R_2,S_2)$ is observable and that $Q^{-1}\allones R_2 \sigma_2$ generates the desired time-varying component of the power demand. Notice that the frequencies of the sinusoidal modes in the power demand have to be the same for all nodes.
As a result, if we consider the contribution of the time-varying component of the disturbance to the optimal steady-state controller, it must be true that
\be\label{optimal.u3}
\overline u_2 =
Q^{-1} \dst\frac{\allones\allones^\trans \Pi_2 \sigma_2}{\allones^\trans Q^{-1}\allones}=
Q^{-1} \allones R_2 \sigma_2,
\ee
where we have exploited the identity $\Pi_2 \sigma_2= Q^{-1}\allones R_2 \sigma_2$. This identity will also be used later in the section.
This characterization points out that, for the existence of a steady state solution with a zero frequency deviation in the presence of time-varying demand, the  exchange of power among the different areas must be constant at steady state and this requires that the intensity of the power demand at one aggregate area should be inversely proportional to the power production cost at the same area. We stress that this is not a limitation of the approach pursued in the paper, but rather a constraint imposed by the model of the power network and the optimal zero frequency regulation problem. We are now ready to state the main result of this section:
\begin{thm}\label{thm3}
Let Assumptions \smod \ref{assum3}, \ref{assum4} and \ref{assum5} \emod hold and suppose that there exists a solution to the regulator equations  (\ref{re-multi-machine-overall2}) with $P^l$ as in \eqref{exosystem.revised}. Then,
given the system
(\ref{multi-machine-overall2}), with exogenous power demand $P^l$ generated by \eqref{exosystem.revised},
with $s_2(\sigma_2)=S_2 \sigma_2$, $S_2$ skew-symmetric and with purely imaginary eigenvalues\footnote{The zero does not belong to the spectrum of $S_2$.}, and $(R_2, S_2)$ an observable pair,  the controllers at the nodes
\be\label{contr.tv}\ba{rcl}
\dot \theta_{1i} &=& \sum_{j\in \mathcal{N}_i^{comm}} ( \theta_{1j}-\theta_{1i}) -q_i^{-1} \omega_i \\
\dot \theta_{2i} &=& S_2\theta_{2i} - q_i^{-1} R_2^T \omega_i \\
u_i &=& q_i^{-1} \theta_{1i}+ q_i^{-1} R_2  \theta_{2i},
\ea\ee
for all $i=1,2,\ldots, n,$ guarantee the solutions to the closed-loop system that start  in a neighborhood of $(\overline \eta, \overline \omega, \overline V, \smodc \overline \theta \emodc)$ to converge asymptotically to the largest invariant set where $\omega_i=0$ for all $i=1,2,\ldots, n$, $\|\nabla_V W_2 \| = 0$ and $u=\overline u$, with $\overline u$ the optimal feedforward input.
\end{thm}

\begin{proof}
We follow the proof of Theorem \ref{th5} {\it mutatis mutandis}. For the sake of generality we continue to use $s_2(\sigma_2)$ instead of $S_2 \sigma_2$, referring to the latter only for those passages in the proof where the linearity of the map $s_2$ simplifies the analysis.
We consider  controllers at the  nodes  of the form \eqref{contr.tv}
where the first term of $u_i$ is inspired by the analogous term in the case of constant power demand (see Theorem \ref{th5}) while the second term is suggested by
\eqref{optimal.u3}.
In stacked form, with $\omega = \mathbf{0}$, the controllers write as
\[\ba{rcl}
\dot \theta_1 &=& -L_{comm} \theta_1 \\
\dot \theta_2 &=& \overline{s}_2 (\theta_2) \\
u &=& Q^{-1} \theta_1+ Q^{-1} (I_n\otimes R_2) \theta_2,
\ea\]
where $\overline{s}_2(\theta)=({s}_2(\theta_{21})^T\ldots {s}_2(\theta_{2n})^T)^T$, $s_2(\cdot)$ is the subvector of $s(\cdot)$ that generates the time-varying component of $\sigma$ and $\theta_2=(\theta_{21}^T\ldots \theta_{2n}^T)^T$
.\footnote{In the case $s_2(\theta_n) = S_2\theta_n$, we have $\overline{s}_2(\theta) = (I_n \otimes S_2)\theta_2$.}
Under appropriate initialization, the system above generates the optimal feedforward input $\overline u$. In fact, if $\overline \theta_1(0)= \frac{\allones\allones^\trans  \Pi_1 \sigma_1(0)}{\allones^\trans Q^{-1}\allones}$, $\overline \theta_2(0)= \allones\otimes \sigma_2(0)$,
then $Q^{-1}\overline  \theta_1+ Q^{-1} (I_n\otimes R_2) \overline \theta_2$, where $\overline  \theta_1, \overline  \theta_2$ satisfy
\[\ba{rcl}
\mathbf{0} &=& -L_{comm} \overline\theta_1\\
\dot {\overline\theta}_2 &=& \overline{s}_2 (\overline\theta_2),
\ea\]
coincides with $\overline u$ defined in \eqref{optimal.u2}.
Following \cite{burger.depersis.aut13}, the stabilizing inputs $v_1$ and $v_2$ are introduced in the controller above to make it incrementally passive. We obtain
\be\ba{rcl}\label{tvinput}
\dot \theta_1 &=& -L_{comm} \theta_1+Q^{-1} v_1\\
\dot \theta_2 &=& \overline{s}_2 (\theta_2) +(I_n\otimes R_2^T) Q^{-1}  v_2\\
u &=& Q^{-1} \theta_1+ Q^{-1} (I_n\otimes R_2) \theta_2.
\ea\ee
The incremental storage function
\[
\Theta(\theta,\overline \theta)=\frac{1}{2} (\theta_1-\overline\theta_1)^T (\theta_1-\overline\theta_1)
+\frac{1}{2} (\theta_2-\overline\theta_2)^T (\theta_2-\overline\theta_2)
\]
satisfies
\be\ba{rll}
\dot \Theta(\theta,\overline \theta)&=&-(\theta_1-\overline\theta_1)
L_{comm}(\theta_1-\overline\theta_1)^T \\&&+(\theta_1-\overline\theta_1)Q^{-1}  v_1 \\&&+
(\theta_2-\overline\theta_2)^T(\overline{s}_2 (\theta_2)-\overline{s}_2 (\overline\theta_2))\\&& +(\theta_2-\overline\theta_2)^T(I_n\otimes R_2^T) Q^{-1}  v_2. \nonumber
\ea\ee
Consider  the incremental storage function
\be\ba{rll} \label{sf_control}
Z(\eta, \overline\eta, \omega, \overline\omega, V, \overline V,\theta,\overline\theta)&=&U(\eta, \overline\eta,\omega, \overline\omega, V, \overline V)\\&&+\Theta(\theta,\overline\theta),
\ea\ee
where $( \overline \eta, \overline V)$ fulfills Assumption \ref{assum3}.
Following the arguments of Lemma \ref{lemHessian}, it is immediate to see that under condition  (\ref{eq.assum3}) 
we have that $\nabla Z|_{\eta= \overline \eta, \omega=\overline \omega, V=\overline V, \theta = \overline \theta} = 0$ and
 $\nabla^2 Z|_{\eta= \overline \eta, \omega=\overline \omega, V=\overline V, \theta = \overline \theta} > 0$,
such that $Z$ has a strict local minimum at  $( \overline \eta, \overline \omega, \overline V, \overline \theta)$.
Under the stabilizing feedback
$
v_1 = -(\omega-\overline \omega),\quad v_2=-(\omega-\overline \omega),
$
the function $Z(\omega, \overline\omega, \eta, \overline\eta, V, \overline V, \theta,\overline\theta)=U(\omega, \overline\omega, \eta, \overline\eta, V, \overline V)+\Theta(\theta,\overline\theta)$ along the solutions to
\be\label{cl.mixed3}
\ba{rcl}
\dot{\eta} &=& D^T  \omega\\
\dot{\overline\eta} &=& \mathbf{0}\\
M\dot{ \omega} &=& -A   \omega - D(\Gamma(\tilde V)\boldsymbol{\sin}(\eta)-\Gamma(\overline V)\boldsymbol{\sin}(\overline\eta)) \\&&+
Q^{-1}(\theta_1-\overline\theta_1)+ Q^{-1}
(I_n\otimes R_2)(\theta_2-\overline\theta_2)
\\
\dot{\overline \omega} &=& \mathbf{0}\\
T\dot{V} &=& -E(\eta)V + \overline E_{fd} \\
\dot{\overline V} &=& \mathbf{0} \\
\dot \theta_1 &=& -L_{comm} \theta_1-Q^{-1} \omega\\
\dot {\overline\theta}_1 &=& \mathbf{0}\\
\dot \theta_2 &=& \overline{s}_2 (\theta_2) -(I_n\otimes R_2^T) Q^{-1}  \omega\\
\dot {\overline\theta}_2 &=& \overline{s}_2 (\overline\theta_2)\nonumber
\ea
\ee
satisfies
\[\ba{rll}
\dot Z &=& -(\omega-\overline\omega)^\trans A (\omega- \overline\omega) -\|\nabla_{V} W_2\|^2_{T^{-1}} \\&&-(\theta_1-\overline\theta_1)^T L_{comm} (\theta_1-\overline\theta_1),
\ea\]
where we have exploited the identities
\be\ba{rll}
u_1-\overline u_1 &=& Q^{-1}(\theta_1-\overline\theta_1)\\
u_2-\overline u_2 &=& Q^{-1}
(I_n\otimes R_2)(\theta_2-\overline\theta_2).\nonumber
\ea\ee
 As $\dot{Z} \leq 0$, one infers convergence to the largest invariant set of points where $\omega=\mathbf{0}$, $\|\nabla_{V} W_2\|=\mathbf{0}$, $\theta_1=\overline \theta_1 +\allones \alpha$, where $\alpha:\R_{\ge 0}\to \R$ is a function.
On the invariant set the dynamics take the form


\be\label{cl.mixed4}
\ba{rcl}
\dot{\eta} &=& \mathbf{0}\\
\mathbf{0} &=&  - D(\Gamma(\tilde V)\boldsymbol{\sin}(\eta)-\Gamma(\overline V)\boldsymbol{\sin}(\overline\eta))\\&&+
Q^{-1}\allones \alpha+ Q^{-1}(I_n\otimes R_2)(\theta_2 - \overline \theta_2 )
\\
\mathbf{0} &=& -E(\eta)\tilde V + \overline E_{fd} \\
\dot{\overline \theta}_1 + \allones \dot \alpha &=& -L_{comm} (\overline \theta_1 + \allones  \alpha)\\
\dot{\theta}_2 - \dot {\overline\theta}_2 &=& \overline{s}_2 (\theta_2 - \overline \theta_2 ),
\ea
\ee
where $\tilde V$ is a constant.  From $\dot{\eta} = \mathbf{0}$ it follows that on the invariant set  $\eta$ is a constant $\eta = \tilde \eta$. From the fourth line in (\ref{cl.mixed4}) we infer that $\alpha$ is a constant. The second line with $\eta = \tilde \eta$ then implies that $q^{-1}_{i}R_2(\theta_{2i} - \overline \theta_{2i} ) = c_{i}$ is a constant as well. Since the term $R_2(\theta_{2i} - \overline \theta_{2i} )$ contains only sinusoidal modes, necessarily $c_{i} = 0$ and from the pair $(R_2, S_2)$ being observable it follows that $\theta_{2i} = \overline \theta_{2i}$. Equal to the proof of Theorem \ref{th5}, pre-multiplying the second line in (\ref{cl.mixed4}) by $\allones^T$ shows that $\alpha=0$ and therefore that $\theta_1 = \overline \theta_1$. We can now conclude that $u_1= \overline u_1$ and $u_2 = \overline u_2$, that is the input $u$ converges to the optimal (time-varying) feedforward input, as claimed. \stopmodif
%
\end{proof}
\subsection{Frequency regulation in the \smod presence \emod of a wider class of time-varying power demand}\label{section6.2}
We continue the previous subsection by considering frequency regulation in the case the power demand is generated by the exosystem
\be\label{exosystem.revised2}
\ba{rll}
\dot \sigma_1 &=&0 \\
\dot \sigma_2 &=& s_2(\sigma_2)\\
\dot \sigma_3 &=& \overline s_3(\sigma_3)\\
P^l &=& \Pi_1 \sigma_1 +Q^{-1}\allones R_2 \sigma_2 + \overline R_3 \sigma_3,
\ea\ee
where additionally to (\ref{exosystem.revised}) we have
$\overline{s}_3(\theta)=({s}_{31}(\theta_{31})^T\ldots {s}_{3n}(\theta_{3n})^T)^T$ and $\overline R_3 = \text{block.diag}(R_{31},\hdots,R_{3n})$. Notice that $s_{3i}(\theta_{3i})$ and $R_{3i}$ can now vary from node to node. As shown in the previous subsection $\overline u$ cannot satisfy (\ref{optimal.u2}) any longer due to the presence of $\sigma_3$. However, compensating for $\sigma_3$ is still a meaningful control task, for otherwise the frequency deviation would not converge to zero any longer. Furthermore, for those cases for which the component $\Pi_1 \sigma_1 +Q^{-1}\allones R_2 \sigma_2$ is much greater in magnitude than  $\overline R_3 \sigma_3$,
$\overline u$ will satisfy (\ref{optimal.u2}) approximately. In order to regulate the frequency deviation to zero when the power demand is generated by (\ref{exosystem.revised2}) we propose controllers inspired by the previous subsection and we adjust the proof of Theorem \ref{thm3} accordingly.

\begin{corollary} \label{coroverall}
Let Assumptions \smod \ref{assum3}, \ref{assum4} and \ref{assum5} \emod hold and suppose that there exists a solution to the regulator equations  (\ref{re-multi-machine-overall2}) with $P^l$ as in \eqref{exosystem.revised2}. Then,
given the system
(\ref{multi-machine-overall2}), with exogenous power demand $P^l$ generated by \eqref{exosystem.revised2}$, s_k(\sigma_k)=S_k \sigma_k$, $S_k$ skew-symmetric and with purely imaginary eigenvalues for $k=2,3$, and
$((R_2\;R_{3i}), \text{block.diag}(S_2,S_{3i}))$ an observable pair,    the controllers at the nodes
\be\label{contr.tvcol}\ba{rcl}
\dot \theta_{1i} &=& \sum_{j\in \mathcal{N}_i^{comm}} (\theta_{1j}-\theta_{1i}) -q_i^{-1} \omega_i \\
\dot \theta_{2i} &=& S_2\theta_{2i} - q_i^{-1} R_2^T \omega_i \\
\dot \theta_{3i} &=& S_{3i}\theta_{3i}  -  R_{3i}^T \omega_i \\
u_i &=& q_i^{-1} \theta_{1i}+ q_i^{-1} R_2  \theta_{2i} + R_{3i}  \theta_{3i},
\ea\ee
for all $i=1,2,\ldots, n,$ guarantee the solutions to the closed-loop system that start in a neighborhood of $(\overline \eta, \overline \omega, \overline V, \overline \eta)$ to converge asymptotically to the largest invariant set where $\omega_i=0$ for all $i=1,2,\ldots, n$, $\|\nabla_V W_2 \| = 0$ and $u=\overline u = \overline u_1 + \overline u_2 + \overline u_3$.
\end{corollary}

\begin{proof}
\startmodif
By adding $\frac{1}{2}(\theta_3 - \overline \theta_3)^T (\theta_3 - \overline \theta_3)$ to the overall storage function (\ref{sf_control}) and following the same lines of reasoning as the proof of Theorem \ref{thm3} we can conclude that the system converges to
the largest invariant set of points where $\omega=\mathbf{0}$, $\|\nabla_{V} W_2\|=\mathbf{0}$, $\theta_1=\overline \theta_1 +\allones \alpha$, where $\alpha:\R_{\ge 0}\to \R$ is a function. On the invariant set the dynamics take the form
\be
\ba{rcl}\label{invcor2}
\dot{\eta} &=& \mathbf{0}\\
\mathbf{0} &=&  - D(\Gamma(\tilde V)\boldsymbol{\sin}(\eta)-\Gamma(\overline V)\boldsymbol{\sin}(\overline\eta))\\&&+
Q^{-1}\allones \alpha+ Q^{-1}(I_n\otimes R_2)(\theta_2 - \overline \theta_2 )\\&&+ \overline R_3(\theta_3 - \overline \theta_3)
\\
\mathbf{0} &=& -E(\eta)\tilde V + \overline E_{fd} \\
\dot{\overline \theta}_1 + \allones \dot \alpha &=& -L_{comm} (\overline \theta_1 + \allones  \alpha)\\
\dot{\theta}_2 - \dot {\overline\theta}_2 &=& \overline{s}_2 (\theta_2 - \overline \theta_2 )\\
\dot{\theta}_3 - \dot {\overline\theta}_3 &=& \overline{s}_3 (\theta_3 - \overline \theta_3)\nonumber
\ea
\ee
from where we can conclude in a similar way as in the proof of Theorem \ref{thm3} that $(\theta_1, \theta_2, \theta_3)$ converges to $(\overline \theta_1, \overline \theta_2, \overline \theta_3)$ and therefore that  $u$ converges to $\overline u$.
\end{proof}
\stopmodif
\smod
\begin{rem}
The focus of this work is on the asymptotic behavior of the system. To obtain a desirable transient response, we can adjust the controller of Corollary \ref{coroverall}, by including additional controller gains $\alpha, \beta_1, \smodc \beta_2  \emodc\in \mathbb{R}_{>0}$, and $\beta_3 \in \mathbb{R}^n_{>0}$ resulting in a controller of the form
\be\ba{rcl}
\dot \theta_1 &=& -\alpha L_{comm} \theta_1 - \beta_1 Q^{-1} \omega\\
\dot \theta_2 &=& \overline s_2 (\theta_2)  - \beta_2 (I_n\otimes R_2^T) Q^{-1}  \omega\\
\dot \theta_3 &=& \overline s_3 (\theta_3)  -  \overline R_3^T \text{diag}(\beta_3) \omega\\
u &=& \beta_1 Q^{-1} \theta_1+ \beta_2 Q^{-1} (I_n\otimes R_2) \theta_2 + \text{diag}(\beta_3) \overline R_3 \theta_3. \nonumber
\ea\ee
The tuning of the various parameters depends on the system at hand and is outside of the scope of this paper. We notice however that the optimality features of the controller are preserved under the addition of the various gains.
\end{rem}

\begin{rem}
Controller (\ref{contr.tv}) (as well as the other controllers introduced in the previous sections) are designed to counteract disturbances generated by exosystems (\ref{exosystem.revised2}). These controllers are also robust to other perturbations. Consider system (\ref{multi-machine-overall2})  with a disturbance $-Q^l$ in addition to $-P^l$. Assume that  $-Q^l$ has a finite $\mathcal{L}_2$-norm, namely $\int_0^{\infty} \|Q^l(s)\|^2 ds<\infty$. In the presence of $Q^l$, the incremental model is modified in such a way that the function $Z(\omega, \overline\omega, \eta, \overline\eta, V, \overline V, \theta,\overline\theta)$ satisfies
\[\ba{rll}
\dot Z &=& -(\omega-\overline\omega)^\trans A (\omega- \overline\omega) -\|\nabla_{V} W_2\|^2_{T^{-1}} \\&&-(\theta_1-\overline\theta_1)^T L_{comm} (\theta_1-\overline\theta_1)-(\omega-\overline\omega)^\trans Q^l.
\ea\]
Further manipulations show that
\[\ba{rll}
\dot Z &\leq& -(\omega-\overline\omega)^\trans \tilde A (\omega- \overline\omega) -\|\nabla_{V} W_2\|^2_{T^{-1}} \\&&-(\theta_1-\overline\theta_1)^T L_{comm} (\theta_1-\overline\theta_1)+\gamma (Q^l)^T Q^l,
\ea\]
where $\gamma= \frac{1}{2 \varepsilon}$, $\tilde A=A - \frac{\varepsilon}{2} I$, and $\varepsilon$ is a number satisfying $0<\varepsilon<2 \min_{i} \{A_i\}$.  Integrating both sides yields
\[
Z(t)-Z(0)\le \gamma \dst\int_{0}^{t} \| Q^l(s)\|^2 ds \le \gamma  \dst\int_{0}^{\infty} \| Q^l(s)\|^2 ds.
\]
This shows that, for initial conditions of the system which are sufficiently close to a strict local minimum of $Z$ and for disturbances  $Q^l$ with a sufficiently small $\mathcal{L}_2$-norm, the solutions of the system remain in a compact level set of $Z$ and as such are bounded and exist for all time. Furthermore,
\be\ba{rll}
 (\omega(t)-\overline\omega)^\trans \tilde A (\omega(t)- \overline\omega)
 \le Z(0) + \gamma \dst\int_{0}^{\infty} \| Q^l(s)\|^2 ds. \nonumber
\ea\ee
Since the left-hand side is bounded for all time,
we have
\[\ba{c}
 \min_{i} \{\tilde A_i\} \sup_{t\ge 0} \|\omega(t)-\overline\omega\|^2\le
 Z(0) + \gamma \dst\int_{0}^{\infty} \| Q^l(s)\|^2 ds,
\ea\]
which shows the existence of  a finite  $\mathcal{L}_2$-to-$\mathcal{L}_\infty$ gain from the disturbance $Q^l$ to the frequency deviation $\omega- \overline\omega$. Similarly, one can show the existence of a finite  $\mathcal{L}_2$-to-$\mathcal{L}_2$ gain.
\end{rem}

\emod
This section contributed to the development of distributed and dynamic controllers based on an internal model design able to generate a time-varying feedforward input such that a zero frequency deviation is obtained in the presence of time-varying power demand. Furthermore we characterized the time-varying power demand that can be compensated optimally under the requirement of zero frequency regulation.
\section{Simulation case study}\label{sec.7}
We illustrate the performance of the controllers on a connected four area network (see \cite{nabavi} how a four area network equivalent can be obtained for the IEEE New England 39-bus system or the South Eastern Australian 59-bus system). This simulation is carried out on the network provided in \cite{li_connecting_2013} and its network topology is shown in Figure 1.
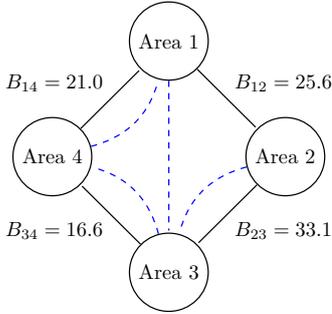
\begin{figure}[h!]
\centering
\resizebox{130pt}{!}{
\begin{tikzpicture}[-,-,shorten >=1pt,auto,node distance=2.8cm,
                    semithick]
  \tikzstyle{every state}=[fill=white,draw=black,text=black]

  \node[state] (A)                    {Area 4};
  \node[state]         (B) [above right of=A] {Area 1};
  \node[state]         (D) [below right of=A] {Area 3};
  \node[state]         (C) [below right of=B] {Area 2};

  \path[] (A) edge              node {$B_{14}=21.0$} (B)
        (B)
            edge              node {$B_{12}=25.6$} (C)
        (C) edge              node {$B_{23}=33.1$} (D)

        (D)
            edge              node {$B_{34}=16.6$} (A);
\path[dashed,blue]
  (A) edge [bend right=30]             node {} (B)
   (B) edge           node {} (D)
        (C) edge [bend right=30]              node {} (D)

        (D)
            edge [bend right=30]              node {} (A);
\end{tikzpicture}}
\caption{A four area equivalent network of the power grid, where $B_{ij}$ denotes the susceptance of the transmission line connecting two areas. The dashed lines represent the communication links.}
\end{figure}
  The values of cost coefficients, generator and transmission line parameters are a slight modification of the ones provided in \cite{li_connecting_2013}, \cite{Ourari2006} and \cite{bergen_power_2000}. An overview of the numerical values of the relevant parameters is provided in Table 2. \begin{table}\label{tab2} \centering
    \begin{tabular}{@{} cl*{4}c @{}}
        & & \rot{Area 1} & \rot{Area 2} & \rot{Area 3} & \rot{Area 4} \\
        \cmidrule{2-6}
  $M_{i}$& & 5.22& 3.98 & 4.49& 4.22\\
  $A_{i}$ && 1.60& 1.22 & 1.38& 1.42\\
  $T_{doi}$ && 5.54& 7.41& 6.11 & 6.22\\
  $X_{di}$ && 1.84& 1.62 & 1.80 & 1.94\\
  $X_{di}^{'}$ && 0.25& 0.17&0.36 & 0.44\\
  $\smodc E_{fdi} \emodc$  && 4.41& 4.20 & 4.37& 4.45\\
  $B_{ii}$ && -49.61& -61.66 & -52.17& -40.18\\
  $\smodc q_i$ && 1.00& 0.75 & 1.50& 0.50 \emodc \\
    \end{tabular}
    \caption{\smodc An overview of the numerical values used in the simulations. System parameters are provided in `per unit', except $T_{doi}$ (seconds). The generation cost coefficient $q_i$ is given in $\frac{\$10^4}{h}$.  \emodc}
\end{table}
\smod The communication among the controllers is depicted in Figure 1 as well and differs from the topology of the power grid. \emod
As a first scenario the power demand $P^l$ \smodc (per unit) \emodc is assumed to be constant and therefore the controller of Section \ref{sec.contr} is applicable. \emodc The system is initially at steady state with a constant load \smodc $P^l(t) = (2.00, 1.00, 1.50, 1.00)^T$,  $t\in [0,10)$  \emodc and according to their cost functions generators take a different share in the power generation such that the total costs are minimized. At timestep \smodc 10 \emodc the load\footnote{ We can interpret an \smod increase \emod of load also as a sudden drop of uncontrollable (renewable) generation. \emod} is increased to \smodc $P^l(t) = (2.20, 1.05, 1.55, 1.10)^T$, $t \geq 10$ \emodc. The frequency response to the control input is given in Figure 2, \smodc where the base frequency is $120 \pi $ rad/s. \emodc From Figure 2 we can see how the frequency drops due to the increased load. Furthermore we note that the controller regulates the power generation such that a new steady state condition is obtained where the frequency deviation is again zero and costs are minimized. \smod To elaborate on this we note that the generation costs at \smodc $t=100$ are $3.36 \times 10^4$ \$/h (assuming a base power of 1000 MVA)\emodc, which is substantially lower than the generation costs when every control area would produce only for its own demand \smodc ($4.79 \times 10^4$ \$/h)\emodc. \emod Since we did not include excitor dynamics in this simulation, the voltages are not regulated. Nevertheless the voltages do not deviate much from their nominal value of 1 per unit. \smod At timestep \smodc 70 \emodc the communication link between area 1 and area 3 is lost, without effecting the frequencies. This is expected since the controller states are already in consensus and it provides numerical evidence that our approach is robust to changes in the communication network as long as the graph stays connected.\emod
 \begin{figure}\label{figth5}
\centering
  \includegraphics[trim= 0cm 8cm 0cm 7.9cm, width=\columnwidth]{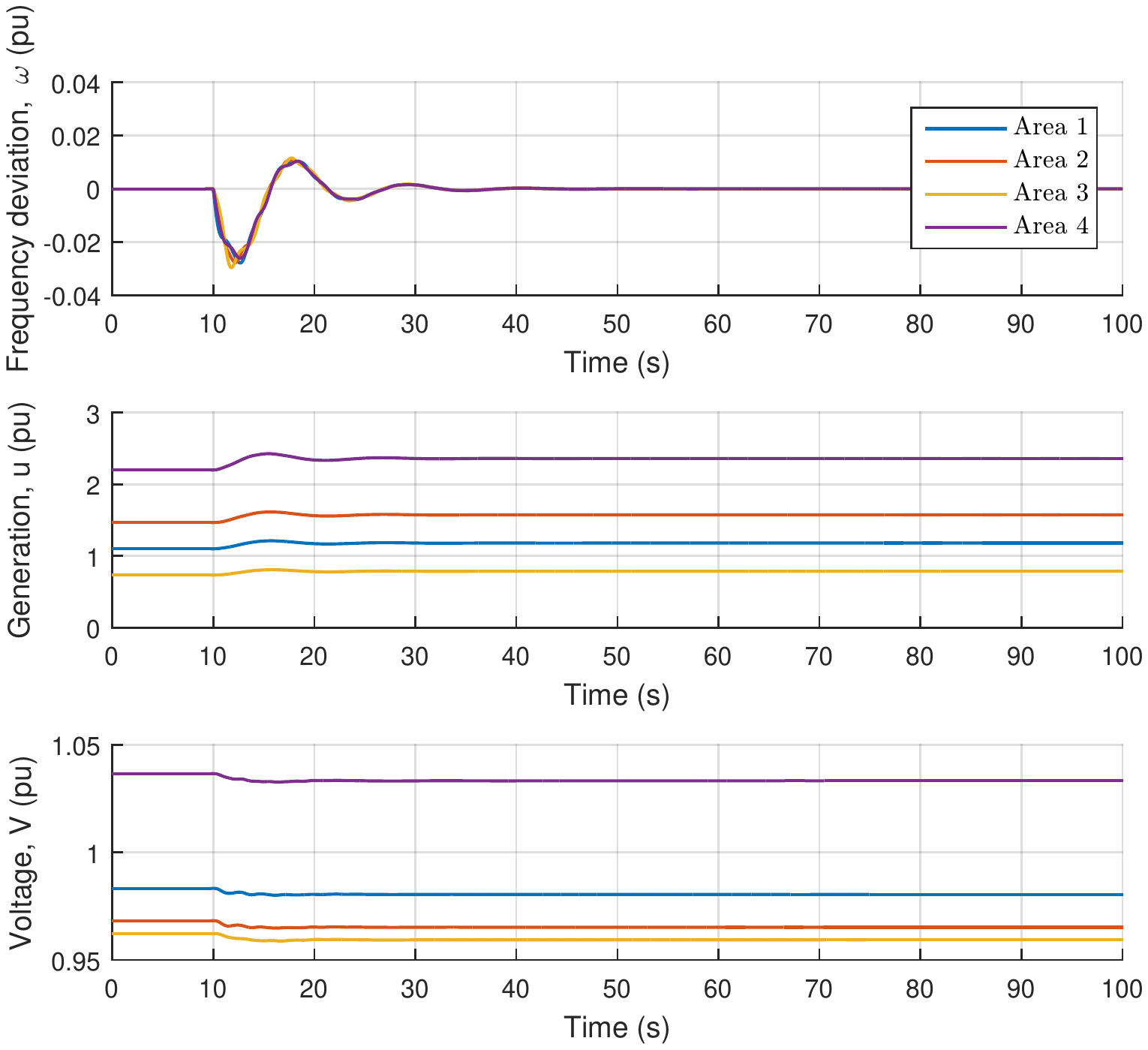}
  \caption{Frequency response and control input using the controller of Section \ref{sec.contr}. The constant load is increased at timestep 10, whereafter the frequency deviation is regulated back to zero and generation costs are minimized.
  }
\end{figure}
\vspace{1em}
As a second scenario we consider a time-varying load as in (\ref{exosystem.revised2}), where the constant demand of scenario 1 is modulated by sinusoidal terms with periods of \smodc 30 \emodc seconds. We note that the controller design does not require that all loads vary with the same frequency and is only assumed here for notational convenience.  The resulting load profile is given by \smodc $P^l(t) = (2.00, 1.00, 1.50, 1.00)^T + 0.040 \times \sin(\frac{2\pi t}{30})(1.10,1.20,0.98,1.00)^T$, $t\in[0,10)$ and $P^l(t) = (2.20, 1.05, 1.55, 1.10)^T + 0.044\times \sin(\frac{2\pi t}{30})(1.04,1.30,0.99,1.00)^T$, $t \geq 10$. \emodc Notice that the sinusoidal term does not belong to $\mathcal{R}(Q^{-1}\allones)$. Accordingly we rely on the controller proposed in Section \ref{section6.2} with $\beta_3 = 0.5\allones$ (see Remark 7) with matrices
\[\ba{ccc}
 S_{3i} = \begin{pmatrix}
  0 & \frac{2 \pi}{30} \\
  \frac{-2 \pi}{30} & 0
\end{pmatrix}, \quad R_{3i} = \begin{pmatrix}
  1 & 0
\end{pmatrix}.
\ea\]
From Figure 3 we can see how the controller provides a time-varying input such that the frequency deviation is driven to zero even in the presence of a time-varying load.
\begin{figure}\label{figth6}
\centering
  \includegraphics[width=\columnwidth, trim= 0cm 8cm 0cm 7cm]{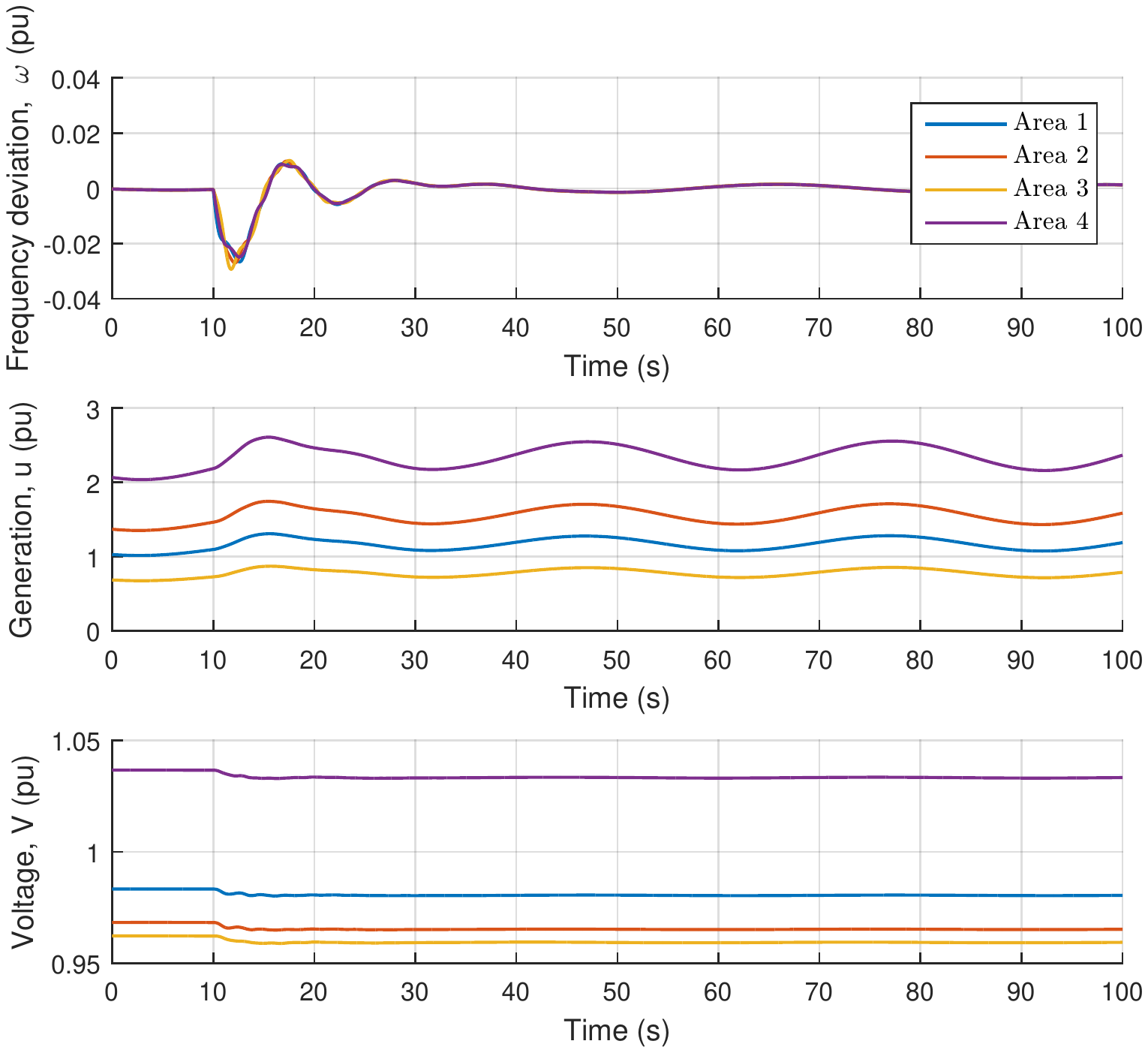}
  \caption{Frequency response and control input using the controller of Section \ref{section6.2}. The time-varying load is increased at timestep 10, whereafter the frequency is regulated back to zero. The costs associated to the constant term of the generation are minimized.}
\end{figure}
Since the time-varying load does not belong to $\mathcal{R}(Q^{-1}\allones)$, the time-varying power generation is not economically optimal anymore. One can check however that the constant term of the generation still converges to the optimum and is equal to the optimal generation in scenario 1. An example of optimal generation in the presence of a time-varying load, where the time-varying load belongs to $\mathcal{R}(Q^{-1}\allones)$, is provided in \cite{burger.et.al.mtns14b} under the assumption of constant voltages.
\section{Conclusions and future work}\label{sec.8}
We have investigated the use of  incremental passivity and internal-model-based controllers to the design of distributed controllers for frequency regulation and cost minimization in power networks.
The approach allows us to consider time-varying power demand. Energy functions which are common in classical literature on power networks have been used as incremental storage functions to analyze and the design the controllers.
%
%
\vspace{1em}

Future work will include the use of more accurate models of the power grid, such as higher order models including e.g.~exciter dynamics. \smodc The inclusion of transmission losses and associated costs will be investigated as well. \emodc There is an increasing attention for models that depart from  the classical swing equations (\cite{zonetti2012}, \cite{caliskan2014}) and interestingly enough these papers use (incremental) passivity arguments for analysis purposes. It is then very natural to wonder whether the methods proposed here can be used to solve demand-supply balancing problems for these more accurate models.
Based on results such as \cite{cox2012}, \cite{Serrani2001} larger classes of time-varying power demand  can be considered, by allowing the use of more general exosystems. 
Finally the analysis has pointed out that asking for a zero frequency deviation in the presence of time-varying power demand restricts the power demand that can be dealt with optimally. As a future research, we will investigate different problems of optimal frequency regulation where the frequency is allowed to differ from the  zero frequency by small variations.


\bibliographystyle{IEEEtran}
\bibliography{ifacbib2}

\begin{thebibliography}{10}
\providecommand{\url}[1]{#1}
\csname url@samestyle\endcsname
\providecommand{\newblock}{\relax}
\providecommand{\bibinfo}[2]{#2}
\providecommand{\BIBentrySTDinterwordspacing}{\spaceskip=0pt\relax}
\providecommand{\BIBentryALTinterwordstretchfactor}{4}
\providecommand{\BIBentryALTinterwordspacing}{\spaceskip=\fontdimen2\font plus
\BIBentryALTinterwordstretchfactor\fontdimen3\font minus
  \fontdimen4\font\relax}
\providecommand{\BIBforeignlanguage}[2]{{%
\expandafter\ifx\csname l@#1\endcsname\relax
\typeout{** WARNING: IEEEtran.bst: No hyphenation pattern has been}%
\typeout{** loaded for the language `#1'. Using the pattern for}%
\typeout{** the default language instead.}%
\else
\language=\csname l@#1\endcsname
\fi
#2}}
\providecommand{\BIBdecl}{\relax}
\BIBdecl

\bibitem{burger.et.al.mtns14b}
M.~B{\"u}rger, C.~{De Persis}, and S.~Trip, ``An internal model approach to
  (optimal) frequency regulation in power grids,'' in \emph{Proc. of the 21th
  International Symposium on Mathematical Theory of Networks and Systems
  (MTNS)}, Groningen, the Netherlands, 2014, pp. 577--583.

\bibitem{Pavlov2008}
A.~Pavlov and L.~Marconi, ``Incremental passivity and output regulation,''
  \emph{Systems and Control Letters}, vol.~57, pp. 400 -- 409, 2008.

\bibitem{burger.depersis.aut13}
M.~B{\"u}rger and C.~{De Persis}, ``Dynamic coupling design for nonlinear
  output agreement and time-varying flow control,'' \emph{Automatica}, vol.~51,
  pp. 210--222, 2015.

\bibitem{chiang.et.al.proc.ieee1995}
H.-D. Chiang, C.-C. Chu, and G.~Cauley, ``Direct stability analysis of electric
  power systems using energy functions: theory, applications, and
  perspective,'' \emph{Proceedings of the IEEE}, vol.~83, no.~11, pp.
  1497--1529, November 1995.

\bibitem{machowski_power_2008}
J.~Machowski, J.~Bialek, and D.~J. Bumby,
  \emph{\BIBforeignlanguage{English}{Power System Dynamics: Stability and
  Control}}, 2nd~ed.\hskip 1em plus 0.5em minus 0.4em\relax Wiley, 2008.

\bibitem{Overview}
I.~Ibraheem, P.~Kumar, and D.~Kothari, ``Recent philosophies of automatic
  generation control strategies in power systems,'' \emph{IEEE Transactions on
  Power Systems}, vol.~20, no.~1, pp. 346--357, 2005.

\bibitem{KTH}
M.~Andreasson, D.~V. Dimarogonas, K.~H. Johansson, and H.~Sandberg,
  ``Distributed vs. centralized power systems frequency control under unknown
  load changes,'' in \emph{Proc. of the 12th European Control Conference
  (ECC)}, Zurich, Switzerland, 2013, pp. 3524--3529.

\bibitem{apos2014}
D.~Apostolopoulou, P.~Sauer, and A.~Domínguez-García, ``Automatic
  generation control and its implementation in real time,'' in \emph{Proc. of
  the 47th Hawaii International Conference on System Sciences (HICSS)}, 2014,
  pp. 2444--2452.

\bibitem{simpson-porco_synchronization_2013}
J.~W. Simpson-Porco, F.~D\"{o}rfler, and F.~Bullo, ``Synchronization and power
  sharing for droop-controlled inverters in islanded microgrids,''
  \emph{Automatica}, vol.~49, no.~9, pp. 2603--2611, 2013.

\bibitem{guerrero2011hierarchical}
J.~M. Guerrero, J.~C. Vasquez, J.~Matas, L.~G. de~Vicu{\~n}a, and M.~Castilla,
  ``Hierarchical control of droop-controlled ac and dc microgrids, a general
  approach toward standardization,'' \emph{IEEE Transactions on Industrial
  Electronics}, vol.~58, no.~1, pp. 158--172, 2011.

\bibitem{dorfler2014breaking}
F.~D{\"o}rfler, J.~W. Simpson-Porco, and F.~Bullo, ``Breaking the hierarchy:
  Distributed control \& economic optimality in microgrids,'' \emph{arXiv
  preprint arXiv:1401.1767 [math.OC]}, 2014.

\bibitem{schiffer_synchronization_????}
J.~Schiffer, R.~Ortega, A.~Astolfi, J.~Raisch, and T.~Sezi, ``Synchronization
  of droop-controlled microgrids with distributed rotational and electronic
  generation,'' \emph{Proc. of the IEEE 52nd Conference on Decision and Control
  (CDC)}, pp. 2334--2339, 2013.

\bibitem{zhang.papa.acc13}
X.~Zhang and A.~Papachristodoulou, ``A real-time control framework for smart
  power networks with star topology,'' in \emph{Proc. of the 2013 American
  Control Conference (ACC)}, Washington, DC, USA, 2013, pp. 5062--5067.

\bibitem{li_connecting_2013}
N.~Li, L.~Chen, C.~Zhao, and S.~Low, ``Connecting automatic generation control
  and economic dispatch from an optimization view,'' in \emph{Proc. of the the
  2014 American Control Conference (ACC)}, Portland, OR, USA, 2014, pp.
  735--740.

\bibitem{zhao_power_2013}
C.~Zhao, U.~Topcu, and S.~Low, ``Power system dynamics as primal-dual algorithm
  for optimal load control,'' \emph{Submitted to {IEEE} Transactions on
  Automatic Control}, 2013.

\bibitem{zhao2014decentralized}
C.~Zhao and S.~Low, ``Decentralized primary frequency control in power
  networks,'' \emph{arXiv preprint arXiv:1403.6046 [cs.SY]}, 2014.

\bibitem{Burger2013b}
M.~B{\"u}rger and C.~{De Persis}, ``Internal models for nonlinear output
  agreement and optimal flow control,'' in \emph{Proc. of 9th IFAC Symposium on
  Nonlinear Control Systems (NOLCOS)}, Toulouse, France, 2013, pp. 289 -- 294.

\bibitem{Wieland2011}
P.~Wieland, R.~Sepulchre, and F.~Allg{\"o}wer, ``An internal model principle is
  necessary and sufficient for linear output synchronization,''
  \emph{Automatica}, vol.~47, pp. 1068 -- 1074, 2011.

\bibitem{Isidori2013}
A.~Isidori, L.~Marconi, and G.~Casadei, ``Robust output synchronization of a
  network of heterogeneous nonlinear agents via nonlinear regulation theory,''
  \emph{IEEE Transactions on Automatic Control}, vol.~59, no.~10, pp.
  2680--2691, 2014.

\bibitem{DePersis2012a}
C.~{De Persis} and B.~Jayawardhana, ``On the internal model principle in the
  coordination of nonlinear systems,'' \emph{IEEE Transactions on Control of
  Network Systems}, vol.~1, no.~3, pp. 272--282, 2014.

\bibitem{cox2012}
N.~Cox, L.~Marconi, and A.~Teel, ``Hybrid internal models for robust spline
  tracking,'' in \emph{Proc. of the 51st IEEE Conference on Decision and
  Control (CDC)}, Maui, HI, USA, 2012, pp. 4877 -- 4882.

\bibitem{Serrani2001}
A.~Serrani, A.~Isidori, and L.~Marconi, ``Semi-global nonlinear output
  regulation with adaptive internal model,'' \emph{IEEE Transactions on
  Automatic Control}, vol.~46, no.~8, pp. 1178--1194, 2001.

\bibitem{burger.et.al.mtns14}
M.~B{\"u}rger, C.~{De Persis}, and F.~Allg\"{o}wer, ``Optimal pricing control
  in distribution networks with time-varying supply and demand,'' in
  \emph{Proc. of the 21st International Symposium on Mathematical Theory of
  Networks and Systems (MTNS)}, Groningen, the Netherlands, 2014, pp. 584--591.

\bibitem{burger.et.al.tcns14}
------, ``Dynamic pricing control for constrained distribution networks with
  storage,'' \emph{IEEE Transactions on Control of Network Systems}, vol.~2,
  no.~1, pp. 88--97, 2014.

\bibitem{zonetti2012}
F.~Shaik, D.~Zonetti, R.~Ortega, J.~Scherpen, and A.~van~der Schaft, ``A
  port-{H}amiltonian approach to power network modeling and analysis,''
  \emph{European Journal of Control}, vol.~19, no.~6, pp. 477 -- 485, 2013.

\bibitem{caliskan2014}
S.~Caliskan and P.~Tabuada, ``Compositional transient stability analysis of
  multimachine power networks,'' \emph{IEEE Transactions on Control of Network
  Systems,}, vol.~1, no.~1, pp. 4--14, 2014.

\bibitem{bai2011cooperative}
H.~Bai, M.~Arcak, and J.~Wen, \emph{Cooperative control design: a systematic,
  passivity-based approach}.\hskip 1em plus 0.5em minus 0.4em\relax Springer,
  2011, vol.~89.

\bibitem{schaft2013}
A.~van~der Schaft and B.~Maschke, ``Port-{H}amiltonian systems on graphs,''
  \emph{SIAM Journal on Control and Optimization}, vol.~51, no.~2, pp.
  906--937, 2013.

\bibitem{bergenTPAS81}
A.~Bergen and D.~Hill, ``A structure preserving model for power system
  stability analysis,'' \emph{IEEE Transactions on Power Apparatus and
  Systems}, vol. PAS-100, pp. 25--35, 1981.

\bibitem{miyagi1104229}
H.~Miyagi and A.~Bergen, ``Stability studies of multimachine power systems with
  the effects of automatic voltage regulators,'' \emph{IEEE Transactions on
  Automatic Control}, vol.~31, no.~3, pp. 210--215, 1986.

\bibitem{chakrabortty2011measurement}
A.~Chakrabortty, J.~H. Chow, and A.~Salazar, ``A measurement-based framework
  for dynamic equivalencing of large power systems using wide-area phasor
  measurements,'' \emph{IEEE Transactions on Smart Grid}, vol.~2, no.~1, pp.
  68--81, 2011.

\bibitem{Ourari2006}
M.~Ourari, L.-A. Dessaint, and V.-Q. Do, ``Dynamic equivalent modeling of large
  power systems using structure preservation technique,'' \emph{IEEE
  Transactions on Power Systems}, vol.~21, no.~3, pp. 1284--1295, 2006.

\bibitem{Zhao2014}
C.~Zhao, U.~Topcu, N.~Li, and S.~Low, ``Design and stability of load-side
  primary frequency control in power systems,'' \emph{IEEE Transactions on
  Automatic Control}, vol.~59, no.~5, pp. 1177--1189, 2014.

\bibitem{dorflersiam}
F.~D\"orfler and F.~Bullo, ``Synchronization and transient stability in power
  networks and nonuniform kuramoto oscillators,'' \emph{SIAM Journal on Control
  and Optimization}, vol.~50, no.~3, pp. 1616--1642, 2012.

\bibitem{Burger2013}
M.~B{\"u}rger, D.~Zelazo, and F.~Allg{\"o}wer, ``Duality and network theory in
  passivity-based cooperative control,'' \emph{Automatica}, vol.~50, no.~8, pp.
  2051 -- 2061, 2013.

\bibitem{chu1999}
C.-C. Chu and H.-D. Chiang, ``Constructing analytical energy functions for
  lossless network-reduction power system models: Framework and new
  developments,'' \emph{Circuits, Systems and Signal Processing}, vol.~18,
  no.~1, pp. 1--16, 1999.

\bibitem{dorflernovel}
F.~D\"orfler and F.~Bullo, ``Novel insights into lossless ac and dc power
  flow,'' in \emph{Proc. of the 2013 Power and Energy Society General Meeting
  (PES)}, Vancouver, Canada, 2013, pp. 1--5.

\bibitem{Plemmons1977175}
R.~Plemmons, ``M-matrix characterizations. i--nonsingular m-matrices,''
  \emph{Linear Algebra and its Applications}, vol.~18, no.~2, pp. 175 -- 188,
  1977.

\bibitem{Aguirre200873}
L.~A. Aguirre, D.~D. Rodrigues, S.~T. Lima, and C.~B. Martinez, ``Dynamical
  prediction and pattern mapping in short-term load forecasting,''
  \emph{International Journal of Electrical Power \& Energy Systems}, vol.~30,
  pp. 73 -- 82, 2008.

\bibitem{Falnes2007185}
J.~Falnes, ``A review of wave-energy extraction,'' \emph{Marine Structures},
  vol.~20, no.~4, pp. 185 -- 201, 2007.

\bibitem{van1957power}
I.~Van~der Hoven, ``Power spectrum of horizontal wind speed in the frequency
  range from 0.0007 to 900 cycles per hour,'' \emph{Journal of Meteorology},
  vol.~14, no.~2, pp. 160--164, 1957.

\bibitem{milan2013turbulent}
P.~Milan, M.~W{\"a}chter, and J.~Peinke, ``Turbulent character of wind
  energy,'' \emph{Physical review letters}, vol. 110, no.~13, p. 138701, 2013.

\bibitem{nabavi}
S.~Nabavi and A.~Chakrabortty, ``Topology identification for dynamic equivalent
  models of large power system networks,'' in \emph{Proc. of the 2013 American
  Control Conference (ACC)}, Washington, DC, USA, 2013.

\bibitem{bergen_power_2000}
A.~R. Bergen and V.~Vittal, \emph{\BIBforeignlanguage{English}{Power systems
  analysis}}, 2nd~ed.\hskip 1em plus 0.5em minus 0.4em\relax Prentice Hall,
  2000.

\end{thebibliography}

%

\end{document}